%% file: main.tex
\documentclass[a4paper,11pt]{amsart}
\usepackage[utf8]{inputenc}
\usepackage{amsmath, amssymb, amsthm, array, footnote, enumerate, enumitem, MnSymbol,  mathtools, multirow}
\usepackage[
backend=biber,
style=alphabetic,
sorting=nyt,
giveninits=true,
maxbibnames=9
]{biblatex}
\renewbibmacro{in:}{}
\usepackage{svg, float, subcaption, bbm, bbold, xcolor}
\usepackage{tikz, tikz-cd}
\usetikzlibrary{matrix, calc, arrows.meta, shapes.geometric, intersections, positioning}

\tikzset{
    every node/.style={},           
    every circle node/.style={circle, minimum size=4mm, inner sep=0pt} 
}
\usepackage{xcolor}
\colorlet{myred}{red!10}
\colorlet{mycyan}{cyan!10}

\textheight9in  \def\DATE{\today}
\hoffset - 2cm  \hbadness=100000

\def\C{{\mathcal{C}}}
\def\S{{\mathbb{S}}}

\def\bbR{{\mathbb R}}
\def\L{{\mathcal{L}}}
\def\hL{{\widehat{\mathcal{L}}}}
\def\P{{\mathcal{P}}}
\def\Q{{\mathcal{Q}}}
\def\R{{\mathbb{R}}}
\def\W{{\mathcal{W}}}
\def\Hom{{\text{Hom}}}
\def\Ob{{\text{Ob}}}
\def\Sets{{\mathtt{Sets}}}
\def\Top{{\mathtt{Top}}}
\def\kVect{{\mathbb{k}\text{-}\mathtt{Vect}}}

\theoremstyle{definition}
\newtheorem{definition}{Definition}
\newtheorem*{theorem*}{Theorem}

\theoremstyle{plain}
{
\newtheorem{lemma}{Lemma}
\newtheorem{proposition}{Proposition}
\newtheorem{theorem}{Theorem}
}
\theoremstyle{definition}
{
\newtheorem{example}{Example}
\newtheorem{corollary}{Corollary}

}


\author[D. Bashkirov]{Denis Bashkirov} 
\address{Czech Academy of Sciences, Institute of Mathematics, {\v Z}itn{\'a} ,25,
         115 67 Prague 1, Czech Republic}

\title{Planar rooted line arrangements and an operad for factorized scattering}

\begin{document}
\begin{abstract}
We introduce two topological non-$\Sigma$ operad structures on planar line arrangements subject to a certain geometric order condition, ensuring a well-defined notion of particle ordering on a distinguished line. This is interpreted in terms of scattering diagrams in purely elastic (1+1)-dimensional theories. We discuss a possible approach to factorized scattering in operadic terms.
\end{abstract}
\maketitle
\section{Introduction}
Consider a system of $n$ identical classical non-relativistic particles of equal mass $m=1$ on a line, governed by the Hamiltonian 
$
H=\frac{1}{2}\sum\limits_{i=1}^{n}p_i^2+g^2\sum\limits_{i<j}V(q_i-q_j),
$
where $q_i$ and $p_i$ denote the position and momentum of the $i$-th particle respectively, $V$ is a pairwise interaction potential and $g$ is a real coupling constant.
The potential $V=V(q)$ is assumed to be repulsive, symmetric about $q=0$, falling off sufficiently rapidly as $|q|\to +\infty$ and impenetrable, meaning $V(q)\to+\infty$ as $|q|\to 0$.
The key examples of such $V$'s include the rational $V(q)=\frac{1}{q^2}$ and the hyperbolic 
$V(q)=\frac{a^2}{\sinh^2\left(aq\right)}$
Calogero-Moser potentials. In the latter case the parameter $a>0$ is the inverse of the characteristic interaction range, and the rational case is restored in the limit $a\to 0$. Both cases are hallmark examples of integrable many-body systems, where integrability can be established by explicitly constructing  a corresponding Lax pair \cite{Moser}\cite{Kulish}\cite[Chapter 3]{Perelomov}. 
This had received further development in Olshanetsky--Perelomov's projection method \cite{OlshanetskyPerelomov} and Kazhdan--Kostant--Sternberg's Hamiltonian reduction \cite{KKS} \cite{Etingof}.
A feature that makes these systems particularly amenable to algebraic and combinatorial analysis is existence of asymptotically-free zones for the particles and a simple asymptotic behavior of solutions. Namely, for $1\leq i\leq n$, the latter are characterized by the asymptotic prescriptions $q_i(t)=p_{i}^{-}t+q_{i}^{-}+o(1)$ for $t\to -\infty$, and $q_i(t)=p_{i}^{+}t+q_{i}^{+}+o(1)$ for $t\to +\infty$ for a set of constants $\{p_{i}^{\pm},q_{i}^{\pm}\}$. Furthermore, by the assumptions on $V$, the particles cannot overtake each other and thus can be indexed in such a way that $q_1(t)<q_2(t)<\dots <q_n(t)$ for all $t$. This implies $p_{1}^{-}>p_{2}^{-}>\dots >p_{n}^{-}$ and 
$p_{1}^{+}<p_{2}^{+}<\dots <p_{n}^{+}$. If, as in the Calogero-Moser case, the Lax matrix can be chosen to admit diagonal asymptotic limits for $t\to\pm \infty$, then the sets of the incoming $\{p_{i}^{-}\}$ and outgoing $\{p_{i}^{+}\}$ momenta can be shown to be equal; in fact, $p_{i}^{-}=p_{n-i+1}^{+}$ for $1\leq i\leq n$.
The $n$ independent integrals of motion can be chosen to be asymptotically symmetric functions of momenta $\sum\limits_{i=1}^{n}p_i^k$ for $1\leq k\leq n$.
One says that such a system supports \emph{purely elastic scattering}. This property is not exclusive to the above model and is known to emerge in other settings as well, including (relativistic) field-theoretic ones, such as soliton scattering in the sine-Gordon model \cite{Zamolodchikov77}\cite{Torielli}, and other quantum integrable systems \cite{Zamolodchikovs},\cite{Dorey} and references therein.

Adhering to the classical Galilean set-up in our exposition, consider an initial ($t=0$) configuration of $n$ particles on a line at positions $q_1<q_2<\dots<q_n$ with incoming asymptotic momenta ${p_1>p_2>\dots>p_n}$. 
The conditions imply that no interaction took place before $t=0$, and each particle will eventually catch up and interact (strictly once) with any other. By equality of the asymptotic momenta sets, every single interaction of two, or more, particles results in a mere exchange of asymptotic momenta between them. 
The evolution of such a system can be geometrically encoded by an arrangement of $n$ lines in a $qt$-plane -- a \emph{scattering diagram}, where the $i$-th line has the slope $p_i$ and passes through the point $t=0$, $q=q_i$.
\begin{center}
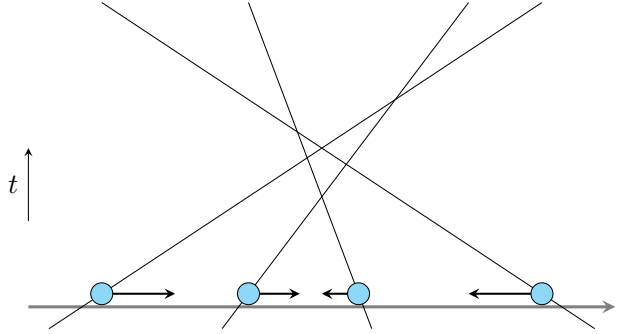

    \resizebox{0.5\textwidth}{!}{\input{sct1}}
    \captionof{figure}{A scattering diagram.}    
\end{center}
Note that while such a line arrangement does contain all the data that determines the corresponding scattering scenario, it does not depict the particles worldlines as such. The lines are the asymptotes, and by our assumptions on the interaction potential, in the vicinity of any point of their intersection, the trajectories of pairwise-interacting particles are disjoint. As an exercise, the reader may trace a possible set of trajectories on the figure above.
\begin{center}
     \includegraphics[width=2.7cm]{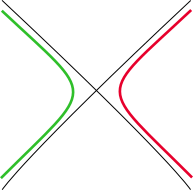}
\end{center}
This is a \emph{reflection} presentation of a scattering process pertaining to our classical picture. A complimentary point of view on combinatorics of particle dynamics is that of the \emph{transmission} presentation. Namely, we may assume that each particle always retains its asymptotic momentum $p_i$, while an interaction amounts to changing the relative order of the particles involved, as if they pass through each other, accompanied by \emph{phase shifts}. 
The latter are discrepancies $\delta_i$ in the asymptotic positions that show up upon identifying the $i$-th incoming particle with the $(n-i+1)$-th outgoing particle in the transmission picture.
Specifically, $\delta_i=q_{n-i+1}^+-q_{i}^-$. Due to $p_{n-i+1}^+=p_i^-(=p_i)$,
$q_i(t)=p_{i}^{-}t+q_{i}^{-}+o(1)$ for $t\to -\infty$ and $q_{n-i+1}(t)=p_{n-i+1}^{+}t+q_{n-i+1}^{+}+o(1)$ for $t\to +\infty$, we have $\delta_i=\lim\limits_{t\to+\infty}[q_{n-i+1}(t)-q_i(-t)-2p_it]$. The values $\delta_1,\dots,\delta_n$ characterize the scattering process and can be thought of as a classical counterpart of a quantum-mechanical $S$-matrix. 
For instance, in the rational Calogero-Moser case all $\delta_i$'s are known to vanish \cite[Chapter 4]{Arutyunov}.
\begin{center}
\begin{tabular}{c c c}
\includegraphics[width=2.6cm]{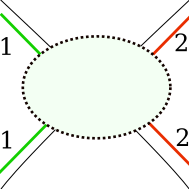}
     &  \quad \quad &
\includegraphics[width=2.6cm]{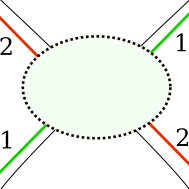}     
\end{tabular}
\captionof{figure}{The reflection and the transmission presentations of an elastic two-body scattering process.}
\end{center}

Integrability of both classical and quantum pure elastic scattering processes is associated with a particular feature that the classical scattering data or the $S$-matrix possess. This is usually stated as the \emph{factorizability} property thereof. Namely, in presence of sufficiently many independent conserved quantities, the complete $S$-matrix accounting for multiparticle collisions can be shown to be representable as a composition of the two-body $S$-matrices. A way to think about it in terms of scattering diagrams is that any such diagram can be perturbed into a line arrangement in a general position, where only pairwise line intersections appear, $\frac{n(n-1)}{2}$ of them in total for $n$ lines. The invariance of a quantum $S$-matrix under such a perturbation is expressed as the Yang-Baxter relation. 
\[
\sum_{a,b,c} S_{ij}^{ba}(\theta_1 - \theta_2) S_{bk}^{nc}(\theta_1 - \theta_3) S_{ac}^{ml}(\theta_2 - \theta_3) =
\sum_{a,b,c} S_{jk}^{bc}(\theta_2 - \theta_3) S_{ic}^{al}(\theta_1 - \theta_3) S_{ab}^{nm}(\theta_1 - \theta_2),
\]
where the spectral parameter $\theta_i$ is identified with the rapidity or momentum of the $i$-th particle; cf. figure \ref{ybconf} below.
Note that in the quantum case the particles may have internal degrees of freedom, thus the $S$-matrix indices above, and come in different flavors. The invariance of the particles count in this setting is not immediate. 

While, effectively, the data of a scattering diagram is that of the initial positions and asymptotic momenta of classical particles, or that of the relative order and asymptotic momenta of quantum ones, the line arrangement interpretation proves itself useful in the following way.


\subsection{An operadic symmetry}
We would like to take an look at the integrability in $(1+1)$-dimensional scattering through the lens of a certain \emph{operadic} symmetry, meaning a consistency condition with respect to an operation of fusing, or composing, multiple initial particle configurations together.
As before, consider the setting of the classical scattering on a line with initial configurations characterized uniquely by the ordered positions $q_1<\dots<q_n$ and momenta $p_1>\dots>p_n$ of the particles. 
A mere set-theoretic union of two initial configurations does not in general result in a well-defined initial configuration due to a possible violation of the order condition on the momenta or possible particle overlaps. 
As a more refined approach, a composition of two particle configurations $a$ and $b$ can be defined by replacing an interval between two adjacent particles of $a$ by an appropriately rescaled, both in terms of positions and momenta, copy of $b$. Note that such a rescaling procedure can be nonlinear and is not unique, since ordered positions and momenta can always be perturbed. The substitution operation $(-\circ-)$ is manifestly non-commutative and requires an additional parameter $i$ -- the insertion position, or the index of the gap between two adjacent particles of $a$, where $b$ is going to be mapped to.
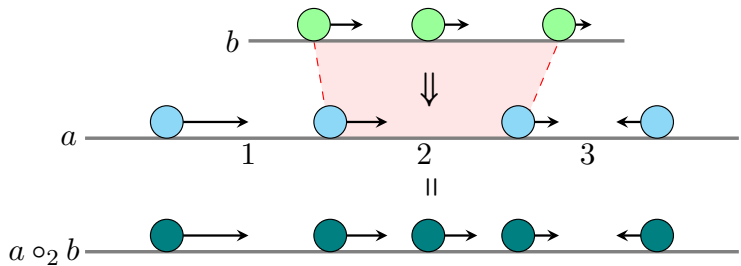
\begin{figure}[H]
    \centering
    \resizebox{0.6\textwidth}{!}{\input{sct2}}
    \caption{Composing two initial particle configurations.}
    \label{CompConf}
\end{figure}
For initial configurations $a$ and $b$ consisting of $m$ and $n$ particles respectively, such an operation is expected to return a valid initial configuration $a\circ_i b$ consisting of $m+n-2$ particles. As noted above, this alone does not determine $a\circ_i b$ uniquely, unless $b$ is a 2-particle configuration. As a possible bootstrap, among all possible ways of defining $(-\circ_i-)$-compositions, we can single out the ones that collectively exhibit a certain ``chronological" symmetry\footnote{One may think of an additional time axis orthogonal to the original spacetime -- the time axis of an agent preparing an initial configuration for us.}.
Namely, a natural condition to impose is independence of an iteratively constructed initial particle configuration on the order in which pairwise {compositions} of three, or more, configurations $a_1,\dots,a_k$ are performed, as long as a chosen relative order of the operands in such a compound composition remains the same. Here, by a relative order on $a_1,\dots,a_k$ we mean a partial order defined by the cover relation $a_i\prec a_j$ iff $a_j$ gets (eventually) substituted into $a_i$. The Hasse diagram of such a poset is a rooted tree, which, upon adding some extra leaves, encodes the corresponding compositional pattern.
\begin{table}[H]
\begin{tabular}{cc}
 \resizebox{10cm}{!}{\input{sct3}} &
 \resizebox{3cm}{!}{\input{comp_tree}} \\ 
\end{tabular}

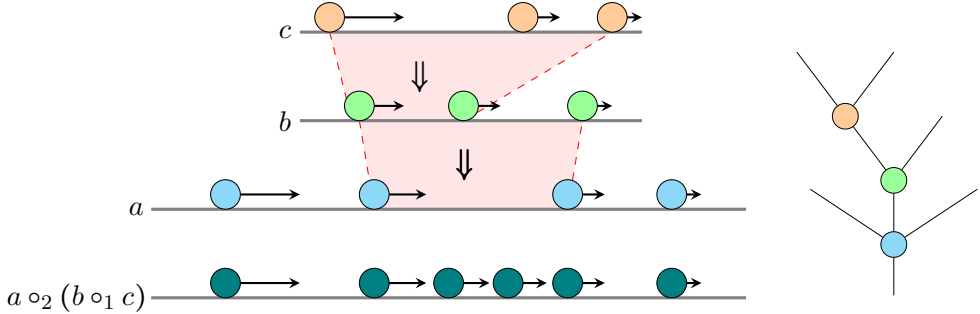
\captionof{figure}{An iterated composition of three particle configurations and the corresponding tree.
Another way to assemble the same configuration is $(a\circ_2 b)\circ_2 c$.}
\end{table}

For three configurations $a$, $b$ and $c$ the discrete ``chronological" invariance condition takes the form
\[
    (a\circ_i b)\circ_j c=a\circ_i(b\circ_{j-i+1}c)   
\]
or
\[
(a\circ_i b)\circ_{j + l - 1} c= (a\circ_j c)\circ_{i} b,
\]
depending on the ranges, where the insertion position indices $i$ and $j$ fall in. In the latter identity, $l+1$ is the total number of particles in $b$.
The invariance conditions for higher compositions can be deduced from here inductively. 
The equations above are a form of generalized associativity formalized by the notion of an operad. 
The affine geometry of line arrangements gets handy when we construct an explicit example of such an operad of particle configurations.
\begin{theorem*}
\emph{
There is a structure of a non-$\Sigma$ operad $\hL$ on scattering diagrams (equivalently, on the initial particle configurations) with binary generators parametrized by \text{rooted line arrangements of rank $3$.}}
\end{theorem*}
The presence of an operad structure on scattering diagrams prompts to bring up its counterpart at the level of $S$-matrices. Indeed, as we show, the $S$-matrix data of
\begin{align*}
    \left\{
    \begin{array}{c}
    \text{Initial/final}\\
    \text{states}\\
    \end{array}
    \right\} \to 
    \left\{
    \begin{array}{c}
    \text{Transition amplitudes}\\
    \text{(classical phase shifts)}      
    \end{array}
    \right\}
\end{align*}
of purely elastic scattering can be recast as a morphism of operads.
Namely one may think of a $C$-\emph{valued} (1+1)-dimensional purely elastic scattering theory as of an operad morphism $S:\hL \to \P$, where $\hL$ is \emph{an} operad of scattering diagrams and $\P$ is a non-$\Sigma$ operad in a symmetric monoidal category $C$. 
In particular, the $S$-matrix factorization property translates into $S$ being entirely determined by its restriction onto the monoid $\hL(1)$, that is, on two crossing lines arrangements.

\section{Preliminaries}
\label{PrelimSec}
\subsection{Operads}
We recall the terminology and the basic concepts of operadic algebra. A point of view that will be convenient for us to maintain throughout the paper is that an \emph{operad} $\P$ is a container of combinatorial objects defined inductively by means of iteratively putting together, or \emph{composing}, some atomic pieces -- the \textit{generators} of $\P$. 
Any generator as well as any composite object $a$ of $\P$ is characterized by a positive integer, the \textit{arity}, that measures receptivity of $a$ to forming new compositions with any given element of $\P$. That is, the arity $m$ of $a\in \P$ counts the number of possible composites $a\circ_i b$ for $i=1,\dots,m$ that $a$ can form for any fixed $b\in \P$. 
In many natural cases, the process of building combinatorial objects via iterated compositions exhibits a certain symmetry
-- a consistency condition of associativity that lies at the core of the definition of an operad.
Furthermore, to enhance a merely combinatorial setting, we would like the elements of $\P$ to have an internal structure (for instance, topological) respected by the composition operations.
Moreover, it is desired to have a well-defined notion of a congruence on an operad that would enable one to impose \textit{relations} among the composites.
A common approach to implement all this is to encapsulate all the combinatorial data of $\P$ into a collection of objects and morphisms of a chosen category $C$ that would host all the generators and relations.

A particular implementation of such a structure, which we find most suitable for our purposes, is that of a \emph{non-$\Sigma$}, or a \emph{non-symmetric}, operad \cite{Giraudo}, \cite[Section 5.9]{LodayVallette}, \cite[Chapter 1]{MarklShniderStasheff}.
Descriptively, in a {non-$\Sigma$} operad $\P$, the elements are modeled on planar\footnote{In the literature, the term \emph{plane tree}, referring to a tree equipped with an embedding into an oriented Euclidean plane considered up to an orientation-preserving isotopy, is used interchangeably. 
Combinatorially, this additional data amounts to defining a cyclic order on the set $F_v$ of all half-edges going out of a vertex $v$ for each vertex $v$ of a given tree.
}
rooted trees, while the operation of forming a composite $a\circ_i b$ corresponds to grafting two planar rooted trees $a, b$ together.
We recall that the grafting operation is defined in the following way. First, by invoking planarity, we enumerate the leaves of all the rooted trees involved. The indexing scheme that we adhere to is defined by traversing any given planar rooted tree in the clockwise caterpillar order, starting at the root, and enumerating all the (non-root) leaves in their order of appearance. Next, upon having all the leaves indexed, the result of \emph{grafting} $b$ onto the $i$-th leaf of $ a $ is defined as the tree $ a \circ_i b $, formed by attaching the root of $ b $ onto the $  $-th leaf of $ a $, while preserving its planar structure, and re-indexing all the leaves of the resulting tree accordingly.
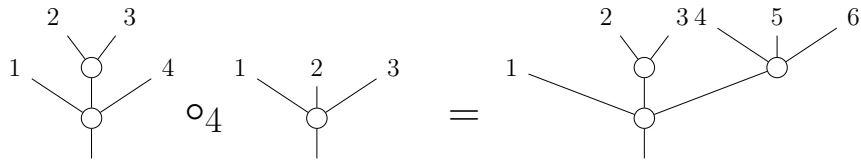
\begin{figure}[H]
    \centering
    \resizebox{0.7\textwidth}{!}{\input{grafting}}
    \caption{Planar rooted trees grafting.}
\end{figure}
In a sense, a non-$\Sigma$ operad is a structure that we get upon allowing all the tree vertices in the above picture to carry some extra data and allowing trees decorated in this manner to be subject to relations expressible in terms of the vertices data and the grafting operation. This is formalized in the following way.
Let $C$ be a category that we assume to be closed symmetric monoidal, complete and cocomplete.
The usual category $\Sets$ of sets, the category $\Top$ of topological spaces and the category ${\mathbb{k}\text{-}\mathtt{(gr)Vect}}$ of (graded) vector spaces over a field $\mathbb{k}$ are the main use-cases for us. Note that each of these categories is concrete, meaning it comes equipped with a (forgetful) faithful functor $C\to \Sets$ of \emph{elements}. Accordingly, in what follows, we will assume, at times, superfluously, that $C$ is concrete.
\begin{definition}
We say that a collection of objects $\P=\{\P(n)\}_{n\geq 1}$ and morphisms $$(-\circ_i-):\P(m)\otimes \P(n)\to \P(m+n-1)$$ in $C$, defined for all $m,n\geq 1$ and $i=1,\dots,m$ is a \emph{non-$\Sigma$ operad in $C$}, provided that these morphisms, called \emph{partial compositions}, satisfy the following conditions.
For any~$k,l,m\geq 1$, $a\in \P(k)$, $b\in \P(l)$, $c\in\P(m)$, we have
\begin{align}
    \label{assoc1}
    (a\circ_i b)\circ_j c=a\circ_i(b\circ_{j-i+1}c)
\end{align}
for all $1\leq i\leq k$, $i\leq j\leq i+l-1$, and
\begin{align}
    \label{assoc2}
    (a\circ_i b)\circ_{j + l - 1} c= (a\circ_j c)\circ_{i} b
\end{align}
for all $1\leq i<j\leq k$.
\end{definition}
The parameter $n$ in $\P=\{\P(n)\}_{n\geq 1}$ is referred to as the \emph{arity}.
Conditions \eqref{assoc1} and \eqref{assoc2} hold for the operation of planar rooted trees grafting, and are called the sequential and the parallel associativity, respectively.
A \emph{morphism} $f: \P \to \Q$ of two non-$\Sigma$ operads in $C$ is a collection of morphisms $f_n:\P(n)\to \Q(n)$ in $C$, defined for all $n\geq 1$, such that $f_{m+n-1}(a\circ_i b)=f_m(a)\circ_i f_n(b)$ for any $a\in \P(m)$, $b\in P(n)$, $m,n\geq 1$ and $1\leq i\leq m$. The non-$\Sigma$ operads in $C$ and their morphisms form a category $nsOp(C)$. There are natural notions of a \emph{suboperad}, one- and two-sided \emph{ideal}, \emph{free operad} and the \emph{quotient operad} extending the corresponding concepts from the realm of monoids and associative algebras. The reader may consult either of the references cited above for a detailed exposition on these notions, if needed.
The term \emph{non-$\Sigma$}, or \emph{non-symmetric}, is meant to indicated the absence of any equivariance conditions imposed on the partial compositions. A comprehensive account of common equivariance relations emerging for operads can be found in \cite{Yau}.
\begin{example}
Consider planar rooted trees with the simplest topology, the oriented paths. On such trees, grafting reduces to oriented paths concatenation $(-\circ_1-)$, which is a monoidal product thereon, and there are no other partial compositions.
\begin{figure}[H]
    \centering
    \resizebox{0.55\textwidth}{!}{\input{concat}}    
\end{figure}
\noindent
More generally, any monoid $A$ in $C$ gives rise to a non-$\Sigma$ operad $\P$ in $C$ \emph{concentrated in arity} $1$. Namely, we set $\P(1):=A$ and $\P(n):=\mathbf{0}$ for $n>1$, where $\mathbf{0}$ is an initial object in $C$. The partial composition $\circ_1:\P(1)\otimes \P(1)\to \P(1)$ is the monoidal product on $A$, while in all other instances, a partial composition $\circ_i:\P(m)\otimes \P(n)\to \P(m+n-1)$ is, necessarily, the identity morphism $id_\mathbf{0}: \mathbf{0} \to \mathbf{0}$. Conversely, given any operad $\P$ in $C$, by \eqref{assoc1}, $\P(1)$ is a semi-group in $C$.
\end{example}

\begin{example}
\label{AsOpEx}
The simplest non-trivial example of a non-$\Sigma$ operad is that of the \emph{associative operad} $As$ in $C$. Namely, for $n\geq 1$, we let $As(n):=\mathbf{1}$, where $\mathbf{1}$ is a (chosen) monoidal unit of $C$. Any partial composition 
$\circ_i:As(m)\otimes As(n)\to As(m+n-1)$
in this operad is the unit multiplication $\mathbf{1}\otimes \mathbf{1} \to \mathbf{1}$ of $C$. 
If a symmetric monoidal category $C$ is \emph{pointed}, meaning $\mathbf{1}$ is a terminal object of $C$, this operad is a terminal object in the category $nsOp(C)$. Thus the qualification the \emph{simplest} given above. The categories $\Sets$ and $\Top$ with their standard symmetric monoidal structures are pointed, while the category ${\mathbb{k}\text{-}\mathtt{Vect}}$ is not.

In terms of trees, $As$ can be described as the operad, where in each arity $n\geq 1$, the component $As(n)$ contains a single tree with $n$ leaves and no decorations. We take such a unique tree representative to be a planar rooted $n$-leaf corolla and define partial compositions as grafting followed by an edge contraction. Identification of two trees upon an edge contraction is an example of a \emph{relation}. As a minimal set of \emph{generators} for this operad we can take a single $2$-leaf corolla, that can be used to produce all higher-arity corollas by iterated self-compositions, and a $1$-leaf corolla that behaves as a neutral element with respect to the $(-\circ_i-)$'s.
\begin{center}
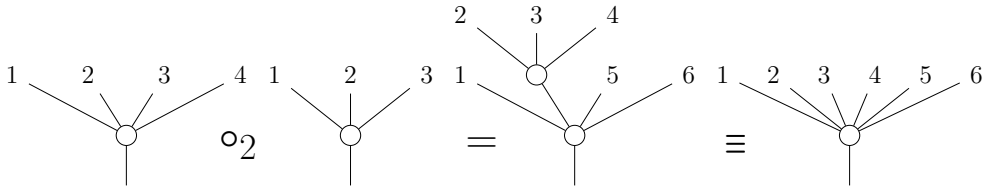

    \resizebox{0.8\textwidth}{!}{\input{ascomp}}
    \captionof{figure}{Evaluating a partial composition in $As$.}
\end{center}
\end{example}

\begin{example}
As a minimal enhancement, the associative operad $As$ defined above can be  given a twist.
Let $C$ be cartesian, and let $A$ be a monoid in $C$.
The \emph{semidirect product} $As\rtimes A$ of the associative operad by a monoid $A$ is known as the \emph{word operad} $\W A$. By definition \cite{Wahl}, 
the elements of $As\rtimes A$ can be identified with planar corollas, whose leaves are decorated by elements of $A$, while partial compositions of $As$ are promoted to be homogeneous with respect to these coefficients.
\begin{center}
    \resizebox{0.28\textwidth}{!}{\input{wordop_corolla}}
\end{center}

By reading the coefficients off the leaves of a decorated corolla in their canonical order, an element of $\W A(m)$ can be identified with a vector, or a \emph{word}, $(a_1,\dots,a_m)\in A^m$ of elements of $A$, that can be thought of as a decoration that we assign to the unique vertex of a $m$-leaf planar rooted corolla.
In terms of words, the partial compositions take the form
\[
(a_1,\dots,a_m)\circ_i(b_1,\dots,b_n)=
(a_1,\dots, a_{i-1}, a_i\cdot b_1, a_i\cdot b_2,\dots, a_i\cdot b_n,a_{i+1},\dots, a_m).
\]
By virtue of various bijections, multiple classical families of combinatorial objects are know to admit a description in terms of the word operad for different choices of $A$ \cite{Giraudo2, Bashk}.
\end{example}

\begin{example}
For any $V\in Ob(C)$, the corresponding \emph{endomorphism} operad $End_V$ in $C$ is defined by setting ${End_V(n):=\underline{\Hom}_C(V^{\otimes n}, V)}$ for each $n\geq 1$, where $\underline{\Hom}_C(-,-)$ denotes the internal hom in $C$. The partial compositions are given by the substitution
\[
(f\circ_i g)(x_1,\dots,x_{m+n-1})=
f(x_1,\dots, x_{i-1},g(x_i,\dots, x_{i+n-1}),x_{i+n},\dots, x_n).
\]
Given a non-$\Sigma$ operad $\P$ in $C$ and $V\in \Ob(C)$, a \emph{representation of $\P$}, or a \emph{$\P$-algebra structure} on $V$, is an operad morphism $\P \to End_V$ in $nsOp(C)$. For example, upon taking $C=\kVect$, a representation of the operad $As$ of example \ref{AsOpEx} on a $\mathbb{k}$-vector space $V$ is equivalent to the data of an associative $\mathbb{k}$-algebra on $V$. Indeed, by definition, for any $n\geq 1$, $As(n)$ is a $1$-dimensional vector space. Hence, a non-$\Sigma$ operad morphism $As\to End_V$ amounts to picking a single $\mathbb{k}$-linear map $\mu_n:V^{\otimes n} \to V$ for each $n\geq 1$ in such a way that a system of defining relations  $\mu_{m}\circ_i\mu_{n}=\mu_{m+n-1}$ of $As$ holds as an equality of $\mathbb{k}$-linear maps on $V^{\otimes m+n-1}$ for all $m,n\geq 1$ and $1\leq i\leq m$. In particular, $\mu:=\mu_2:V\otimes V\to V$ is a bilinear product subject to the edge-contraction relation
\[
\mu\circ_1 \mu=\mu_3=\mu\circ_2 \mu
\]
expressing the usual associativity condition.
\end{example}

\begin{example}
For $n\geq 1$, let $PRT(n)$ be the set of all (isomorphism classes of) planar rooted trees with $n$ leaves. Tautologically, the tree grafting operation turns the collection $PRT:=\{PRT(n)\}_{n\geq 1}$ into a non-$\Sigma$ operad. For any $k\geq 1$, the subcollection $PRT_k:=\{PRT_k(n)\}_{n\geq 1}$ consisting of $k$-ary planar rooted trees forms a suboperad of $PRT$. Each of these suboperads $PRT_k$ is generated by a single element, the $k$-leaf corolla.
\end{example}

\begin{example}
\label{LDEx}
For $n\geq 1$, let $\mathcal{D}_1(n)$ be the set of all collections of $n$ closed non-trivial subintervals, with pairwise disjoint interiors, embedded into the unit interval $I=[0,1]$. Such a set can be given a natural topology as a subspace of $[0,1]^{2n}$, and subintervals within any configuration $a\in \mathcal{D}_1(n)$ can be enumerated naturally in their order of appearance as $I$ is traced from $0$ to $1$.
For any two subinterval configurations $a\in \mathcal{D}_1(m)$, $b\in \mathcal{D}_1(n)$ and $1\leq i\leq m$, the partial composition $a\circ_i b$ is defined by mapping the unit interval containing configuration $b$ onto the $i$-th subinterval of $a$ via an orientation-preserving affine map.
This results in a configuration of $m+n-1$ subintervals in $[0,1]$ and, in fact,
turns $\mathcal{D}_1=\{\mathcal{D}_1(n)\}_{n\geq 1}$ into an operad in $\Top$, known as the (non-$\Sigma$) operad of \emph{little $1$-disks}.
This operad can be regarded as the topological version of the associative operad~$As$. 
Indeed, for each $ n \geq 1 $, any two subinterval configurations in $\mathcal{D}_1(n)$ can be homotopically transformed to each other, and the subintervals enumeration in any partial composition $a\circ_i b$ matches the one from the definition of $As$ in example \ref{AsOpEx}.
This equivalence identifies $\pi_0(\mathcal{D}_1)$ with $ As $ viewed as non-$\Sigma$ operads in $\Sets$.
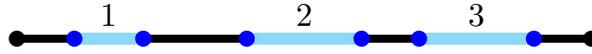
\begin{figure}[H]
    \centering
    \resizebox{0.5\textwidth}{!}{\input{d1}}
    \caption{An element of $\mathcal{D}_1(3)$.}
\end{figure}
\noindent
To familiarize oneself with partial composition of little disks and the theme of recursion frequently emerging in the context of operads, an interested reader might consider the suboperad of $\mathcal{D}_1$ generated by a certain element of arity~$2$ -- the configuration $[0,1/3]\cup[2/3,1]$. This is a non-$\Sigma$ operad of pre-Cantor sets, isomorphic to $PRT_2$. The Cantor set itself can be used as a left module generator over this non-$\Sigma$ operad in $\Top$.

Another suboperad of $\mathcal{D}_1$ that will be particularly useful for us by virtue of encoding static particle configurations on a line is a special case of a tiling operad that we discuss below.
\end{example}

\section{Operads of tilings and configuration spaces.}
\label{TilingEx}
Self-similar tilings of polyhedra (\textit{reptiles} and \emph{irreptiles} \cite{Reid}\cite{Golomb}), provide natural examples of combinatorial operads. For instance, consider tilings of a L-shaped tromino by homothetically rescaled smaller copies thereof.
\begin{table}[H]
\centering
\begin{tabular}{cc}
\resizebox{2.3cm}{!}{\includegraphics{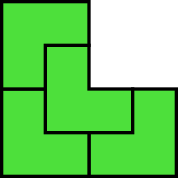}} & 
\resizebox{2.8cm}{!}{\includegraphics{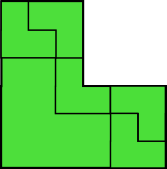}}\\
\end{tabular}
\captionof{figure}{Two self-similar L-shaped tilings.}
\label{lshape}
\end{table}
The arity of a tiling is defined as the number of elementary tiles contained within. A partial composition $a\circ_i b$ of a tiling $a$ of arity $m$ and a tiling $b$ of arity $n$ amounts to substituting a homothetically rescaled, translated and, possibly, rotated copy of $b$ for the $i$-th tile of $a$.  This predicates on the existence of a certain tile-enumerating scheme that allows one to address each tile by its unique index $i$. 
We define such a scheme inductively in the following way.
For any tile configuration that is indecomposable in the sense of the $\circ_i$-products, the enumeration of the tiles is a choice to be made (in general, this could be subject to a consistency condition across the indecomposables if there are relations between them). Next, for tilings $a$, $b$ of arities $m$ and $n$ respectively, and any $1\leq i\leq m$, all the tiles of $a$ with indicies $j<i$ retain their respective indicies in $a\circ_i b$. Any tile of $a$ of index $j>i$ changes it to $j':=j+n-1$. 
All the tiles of $b$, originally indexed $1,2,\dots,n$, acquire indicies $i, i+1,\dots, i+n-1$ respectively in $a\circ_i b$; cf. leaves enumeration in $a\circ_i b$ in example \ref{AsOpEx}.

A tiling pattern shown on the left of figure \ref{lshape} (\emph{Golomb's configuration}), can be used as a single generator of arity four for a family of self-similar tromino tilings forming a non-$\Sigma$ operad (in $\Sets$).
The enumeration of four tiles within the generator is arbitrary. The enumeration of tiles in all composite configurations follows the scheme described above.
The non-$\Sigma$ operad generated in this manner is isomorphic to the operad of $4$-ary planar rooted trees $PRT_4$. A concrete isomorphism depends on the choice of the tiles enumeration in the generator.
The figure below depicts such a correspondence based on the enumeration of the generator tiles as shown on the right.
\begin{table}[H]
\begin{tabular}{ccc}
 \resizebox{3.5cm}{!}{\includegraphics{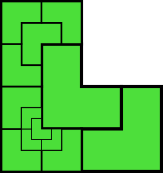}} &
 \resizebox{5.5cm}{!}{\input{lshape_tree}} &
 \resizebox{2cm}{!}{\includegraphics{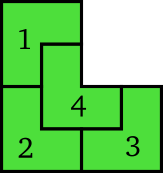}}
\end{tabular}
\captionof{figure}{A tiling generated by Golomb's configuration and the corresponding tree. On the right: the choice of the generator that induces the corresponding operadic bijection.}
\end{table}
\noindent
Note that this is a proper suboperad of the operad of all possible self-similar tromino tilings. For instance, the tiling of arity $6$ shown on the right of figure \ref{lshape}\cite{Khovanova} is not a composite of Golomb's configurations, and is, in fact, $\circ_i$-indecomposable. 

As a geometric model, a self-similar tiling operad can be presented as a suboperad of the word operad $\W G$, where $G\leq Aut(X)$ is the group of the admissible tile transformations and $X$ is the corresponding ambient space. For the above example of trominos embedded in a Euclidean plane, $G$ can be taken to be the orientation-preserving affine similarity group -- the subgroup of the affine group 
$Aff(\R^2)$ restricted to scaling by $s>0$, translations and rotations.
To this end, we fix a reference copy of a basic tile $T$ in $X$, such as a single tromino in a plane as in the above example. We will assume that $T$ has no non-trivial automorphisms in $G$, by putting some framing on $T$, if needed. 
Then any tiling of $T$ consisting of pieces $T_1,T_2,\dots,T_n$ can be uniquely encoded by the word $(g_1,g_2,\dots,g_n)\in G^n$, where each $g_i$ satisfies $g_iT=T_i$. 
These are subject to $\bigcup\limits_{i=1}^{n} T_i=T$ and $int(T_i)\cap int(T_j)=\varnothing$ for $i\neq j$. Note that the tiles are enumerated, all $g_i$'s are distinct, and $|\S_n|$ ways of relabeling the tiles correspond to all the letter permutations within the word.
\begin{lemma}
\label{TilingOpProp}
    The family $\mathcal{T}=\{T(n)\}_{n\geq1}$, where $T(n)$ consists of all tiling words $(g_1,g_2,\dots,g_n)\in G^n$ as above, is a non-$\Sigma$ suboperad $\W G$. of  It is (non-canonically) isomorphic to the operad of labeled self-similar tilings of $T$.
\end{lemma}
\begin{proof}
For any ${(g_1,\dots,g_m)\in\mathcal{T}(m)}$,
${(h_1,\dots,h_n)\in\mathcal{T}(n)}$ and $1\leq i\leq m$,
we have
\begin{align}
\label{WGcomp}
(g_1,\dots,g_m)\circ_i(h_1,\dots,h_n)=(g_1,\dots,g_{i-1},g_ih_1,\dots g_ih_n,g_{i+1},\dots, g_m).
\end{align}
Multiplying $\bigcup\limits_{i=1}^{n} h_i T=T$ on the left by $g_i$ we deduce that the subword $g_{i-1},g_ih_1,\dots g_ih_n$ yields a valid tiling of $g_i T$. Thus, \eqref{WGcomp} encodes a valid tiling of $T$. All the topological conditions are met due to $G$ acting homeomorphically on $X$.
\end{proof}
Non-canonicity is due to making a choice of a template tile $T\subset X$. For any other choice $gT$, the tiling operad is $g\mathcal{T}g^{-1}=\{gT(n)g^{-1}\}_{n\geq1}$, where $G$ acts on each $T(n)$ diagonally.

\subsection{Operads of static particle configurations}
Consider a particular one-dimensional instance of the self-similar tiling operad. Namely, let
$\C$ be the operad of tilings of the closed unit interval $I=[0,1]$ by subintervals.
Equivalently, $\mathcal{C}$ is the suboperad of the little $1$-disks operad $\mathcal{D}_1$ of example \ref{LDEx} that consists of all maximally stretched-out subinterval configurations in $[0,1]$. These are collections of closed non-degenerate subintervals $I_1,\dots,I_n\subseteq I$ with pairwise disjoint interiors, such that $\bigcup\limits_{i=1}^{n} I_i=I$.
As a topological space, $\mathcal{C}(n)$ can be identified with the configuration space $UConf_{n-1}(0,1)$ of $n-1$ unlabeled points in $(0,1)$, where pairwise-adjacent intervals $I_i, I_{i+1}$ meet.
\begin{figure}[H]
    \centering
    \resizebox{0.5\textwidth}{!}{\input{conf}}
    \caption{An element of $UConf_2(0,1)\simeq \mathcal{C}(3)$.}
\end{figure}
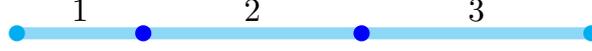
\noindent
The family of configuration spaces $UConf(0,1) = \{UConf_{n-1}(0,1)\}_{n \geq 1} $ acquires the structure of a non-$\Sigma$ operad as a suboperad of the little $1$-disks operad through a bijective correspondence with $\mathcal{C}$.
\begin{figure}[H]
    \centering
    \resizebox{0.85\textwidth}{!}{\input{conf2}}
    \caption{Evaluating a partial composition in $\mathcal{C}$.}
\end{figure}
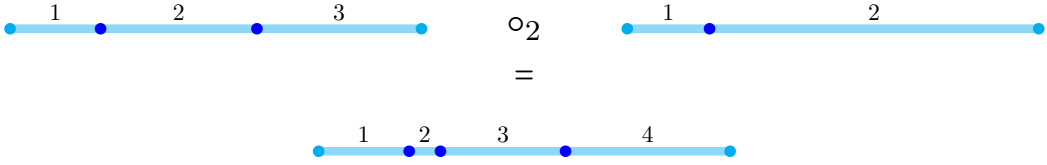
\noindent
Note that a configuration of $n-1$ points in $(0,1)$ is an element of arity $n$ of this operad.
Furthermore, despite being unlabeled, the points come with a canonical enumeration induced by the usual order on $[0,1]$.
The lengths $\lambda_1,\dots,\lambda_n$ of the consecutive tiling intervals $I_1,\dots,I_n$ define a coordinate system on $UConf_{n-1}(0,1)$, which identifies the latter with the interior $int(\Delta^{n-1})$ of the standard geometric $(n-1)$-simplex. In fact, it identifies $UConf_{n-1}(0,1)$ with a suboperad the word operad $\W[0,1]$, where $[0,1]$ with its multiplicative structure is treated as a topological  monoid. The partially-defined addition on $[0,1]$ allows one to cut $UConf_{n-1}(0,1)$ out of the cube $\W[0,1](n):=[0,1]^n$ by means of an affine equation $\lambda_1+\dots+\lambda_n=1$ with $\lambda_i>0$ for all $1\leq i\leq n$. 
Non-$\Sigma$ operad $\C$ admits a family of binary generators $P_\lambda$, $0<\lambda<1$. In terms of interval tilings, $P_\lambda$'s are two-interval configurations of the form $[0,\lambda]\cup[\lambda, 1]$, or equivalently, a point $0<\lambda<1$ in $UConf_1(0,1)$.
These generators are subject to the barycentric relations 
\begin{align}
\label{Bary}
P_{\alpha+\beta}\circ_1 P_{\frac{\alpha}{\alpha+\beta}}=
P_{\alpha}\circ_2 P_{\frac{\beta}{1-\alpha}}
\end{align}
for any $\alpha,\beta>0$, $\alpha+\beta<1$. Any convex set $X\subset \R^N$ is a $\C$-algebra in $\Top$.

For our purposes, we need an operad structure for particle configurations on the entire line $\R$. One way to get it is to transfer an operad structure from $UConf(0,1)$ to $UConf(\R)$ along a homeomorphism $(0,1)\simeq \R$. On the other hand, a more natural approach that circumvents having to choose such a homeomorphism is to use the specifics of the line geometry. Namely, since any finite point configuration $Q$ in $\R$ is contained in an interval $[Q_1,Q_n]$, where $Q_1$ and $Q_n$ are the well-defined leftmost and rightmost points of $Q$ respectively, there is a unique orientation-preserving affine map sending $[Q_1,Q_n]$ to any interval $[P_1,P_m]$ by mapping $Q_1$ to $P_1$ and $Q_n$ to $P_m$; cf. figure \ref{CompConf}.

Proposition \ref{TilingOpProp} provides a way to formalize it. This model serves as a precursor to the main construction of section \ref{LinArrSec}.
To this end, let $G$ be the orientation-preserving affine group $Aff_+(\R^1)$.
Then a tiling of $I=[0,1]$ by subintervals $I_1,\dots,I_n$ of lengths $\lambda_1,\dots,\lambda_n$ can be identified with the tuple
\begin{align}
\label{AffWord}
\left(
\left[
\begin{array}{cc}
\lambda_1 & b_1 \\
0   &  1
\end{array}
\right],\dots,
\left[
\begin{array}{cc}
\lambda_n & b_n \\
0   &  1
\end{array}
\right]
\right)\in G^n
\end{align}
subject to the boundary matching and the normalization conditions that read
\begin{align}
    \label{MatchCond}
    \left[
    \begin{array}{cc}
    \lambda_i & b_i \\
    0   &  1
    \end{array}
    \right]
    \left[
    \begin{array}{c}
     1\\
     1 
    \end{array}
    \right]=
    \left[
    \begin{array}{cc}
    \lambda_{i+1} & b_{i+1} \\
    0   &  1
    \end{array}
    \right]
    \left[
    \begin{array}{c}
     0\\
     1 
\end{array}
\right],\quad 1\leq i\leq n-1
\end{align}
and
\begin{align}
    \label{NormCond}
    \left[
    \begin{array}{cc}
    \lambda_{1} & b_{1} \\
    0   &  1
    \end{array}
    \right]
    \left[
    \begin{array}{c}
     0\\
     1 
    \end{array}
    \right]=
    \left[
    \begin{array}{c}
     0\\
     1 
    \end{array}
    \right], \quad 
    \left[
    \begin{array}{cc}
    \lambda_{n} & b_{n} \\
    0   &  1
    \end{array}
    \right]
    \left[
    \begin{array}{c}
     1\\
     1 
    \end{array}
    \right]=
    \left[
    \begin{array}{c}
     1\\
     1 
    \end{array}
    \right]
\end{align}
respectively.
Lemma \ref{TilingOpProp} identifies $\mathcal{C}$
with the suboperad of $\W G$ consisting of all such tuples.

Upon omitting the normalization condition \eqref{NormCond},
a tuple \eqref{AffWord} can be identified with a 
tiling of an arbitrary interval $[b_1,b_n+\lambda_n]\subset \R$ by $n$ subintervals $[b_1,b_2],\dots,[b_{n-1},b_n],[b_n,b_n+\lambda_n]$. Equivalently, this is a configuration of $n+1$ points $b_1,b_2,\dots,b_n,b_n+\lambda_n$ in $\R$.
At the same time, the partial compositions $(-\circ_i-)$ of the word operad $\W G$ no longer descent onto the set of such tuples. Indeed, as an example, both
$
a=
\left(
\left[
    \begin{array}{cc}
    1 & 0 \\
    0   &  1
    \end{array}
    \right],
   \left[
    \begin{array}{cc}
    1 & 1 \\
    0   &  1
    \end{array}
    \right]
\right)
$ 
and
$
b=
\left(
\left[
    \begin{array}{cc}
    1 & 1 \\
    0   &  1
    \end{array}
    \right],
   \left[
    \begin{array}{cc}
    1 & 2 \\
    0   &  1
    \end{array}
    \right]
\right)
$ 
satisfy the matching condition, but do not satisfy the normalization one. Upon computing
$a\circ_1 b=\left(
\left[
    \begin{array}{cc}
    1 & 1 \\
    0   &  1
    \end{array}
    \right],
   \left[
    \begin{array}{cc}
    1 & 2 \\
    0   &  1
    \end{array}
    \right],
 \left[
    \begin{array}{cc}
    1 & 1 \\
    0   &  1
    \end{array}
    \right]
\right)$
we find out that it is not a valid tiling word. Indeed, the matching condition for the $2$nd and $3$rd entries does not hold (also, can be seen from having two equal entries).

The partial compositions are to be modified, and
this can be done by employing a variation of the moving frame method \cite{Olver}.
\begin{proposition}
    \label{MovingFrameLemma}
    Let $\Q=\{\Q(n)\}_{n\geq 1}$ be a collection of objects of a category $C$ with a group $G$ acting on each component $\Q(n)$, $\rho=\{\rho_n:\Q(n)\to G\}$
    be a family of equivariant maps, $\rho_n(g\cdot z)=g\cdot\rho_n(z)$ for any $n\geq 1$, $z\in \Q(n)$, $g\in G$.    
    For $n\geq 1$, let $\P(n):=\{z\in \Q(n)|\rho(z)=e\}$
    If $\P=\{\P(n)\}_{n\geq 1}$ 
    has the structure of a non-$\Sigma$ operad, this structure can be lifted to
    one on $\Q$. Namely, for $a\in\Q(m)$, $b\in \Q(n)$, $1\leq i\leq m$, we set
    \begin{align}
    \label{MovingFrameComp}
    a *_i b := \rho(a)[\rho(a)^{-1}a\circ_i \rho(b)^{-1}b].
    \end{align}    
\end{proposition}
In words, we evaluate a composition within a distinguished reference frame associated with $\P$ and then state  the result in the reference frame of $a$ in $\Q$.
\begin{proof}
First, note that $a*_i b$ is well-defined. Indeed, for any $z\in \Q$, we have $\rho(\rho(z)^{-1}z)=\rho(z)^{-1}\rho(z)=e$.
Thus, $\rho(z)^{-1}z\in \P$, and the partial composition in the brackets is well-defined.
Next, observe that
\[
\rho(a*_i b)=\rho(\rho(a)[\rho(a)^{-1}a\circ_i \rho(b)^{-1}b])
=\rho(a)\rho(\underbrace{\rho(a)^{-1}a\circ_i \rho(b)^{-1}b)}_{\in \P})=\rho(a).
\]
The last equality is due to $\rho(z)=e$ for any $z\in\P$.

Now, let $a,b,c\in \Q$ and $i,j$ be as in the hypothesis of the sequential associativity axiom \eqref{assoc1}.
Then
\begin{align*}
    (a*_i b)*_{j} c&=    
    \rho(a*_i b)[\rho(a*_i b)^{-1}(a *_i b)\circ_{j} \rho(c)^{-1}c]
    =    \rho(a)[\rho(a)^{-1}(a *_i b)\circ_{j}\rho(c)^{-1}c]\\
    &=
    \rho(a)[\underbrace{(\rho(a)^{-1}a \circ_i \rho(b)^{-1}b)\circ_{j}\rho(c)^{-1}c}_{\in \P}]
    \overset{\eqref{assoc1}}{=}
    \rho(a)[\rho(a)^{-1}a \circ_i (\rho(b)^{-1}b\circ_{j-i+1}\rho(c)^{-1}c)]\\
    &=
    \rho(a)[\rho(a)^{-1}a \circ_i \rho(b)^{-1}(b *_{j-i+1}c)]
    =
    \rho(a)[\rho(a)^{-1}a \circ_i \rho(b *_{j-i+1} c)^{-1}(b *_{j-i+1}c)]=a*_i(b*_{j-i+1}c).
\end{align*}
The proof of parallel associativity is analogous and differs only by applying \eqref{assoc2} instead of \eqref{assoc1} in the second line above.
\end{proof}
In our case, we take $\Q(n)$ to be the set of all tuples \eqref{AffWord} subject to the boundary matching condition \eqref{MatchCond} for all $n\geq 1$.
The set is taken with the diagonal left action by $G=Aff_+(\R^1)$. 
The normalization map $z\mapsto \rho(z)^{-1}z$ is defined by taking
$\rho_n(z):=
\left[
\begin{array}{cc}
   \lambda_1+\dots+\lambda_n  & b_1 \\
    0 & 1 
\end{array}
\right]
$, where $z$ is as in \eqref{AffWord}.
The corresponding cross-section operad $\P$ is the operad $\mathcal{C}$ from above. Note that points $0$ and $1$ are no longer distinguished in $\Q$. For this reason, any elements of $\Q(n)$ is identified with a configuration of $n+1$ points in $\R$. In what follows, we will denote this non-$\Sigma$ operad $\Q$ by $\widehat{\mathcal{C}}$.
\begin{figure}[H]
    \centering
    \resizebox{0.85\textwidth}{!}{\input{conf3}}
    \caption{Evaluating a partial composition in $\widehat{\mathcal{C}}$.}
\end{figure}
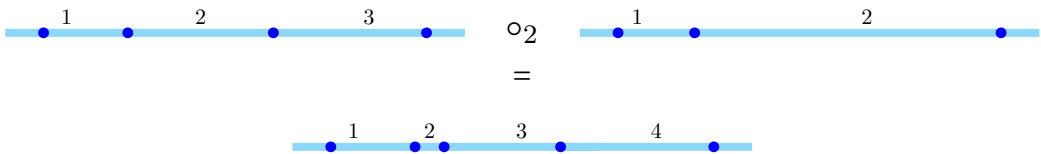

As an upshot, starting with an non-$\Sigma$ operad $\P$ and a group $G$, there is a way to extend $\P$ to a larger $G$-equivariant operad $\Q$. Such an extension is parametrized by a choice of a moving frame ${\rho=\{\rho_{n}:\Q(n)\to G\}_{n\geq 1}}$. This will be used again in the next section.
\section{An operad of planar rooted line arrangements}
\label{LinArrSec}
Let $L$ be a collection of $n+1$ lines $\ell_0,\ell_1,\dots,\ell_n$ embedded into an oriented Euclidean plane $\bbR^2$. Assume that line $\ell_0$ is given a direction. The open half-plane that lies to the left of $\ell_0$ with respect to this direction will be referred to as the  \textit{upper} half-plane of $L$.
\begin{definition}
\label{PRLDef}
We say that a collection of lines $L=\{\ell_0,\ell_1,\dots,\ell_n\}$ in $\bbR^2$ is a \emph{planar rooted line arrangement} of rank $n$ if
\begin{enumerate}
\item
no two lines among $\ell_0,\ell_1,\dots,\ell_n$ are parallel or overlapping;
\item 
\label{PRLDef2}
for any two lines $\ell_i$, $\ell_j$, such that $i,j\neq 0$ and $i\neq j$, the point of intersection $\ell_i\cap \ell_j$ lies in the upper half-plane;
\item 
the order in which the points of intersection $L_i:=\ell_0\cap \ell_i$, for $i\geq 1$, appear on the oriented line $\ell_0$ as we trace it in the positive direction agrees with the natural order on the line indices $1,\dots,n$. \end{enumerate}
The distinguished line $\ell_0$ will be called the \textit{root} of an arrangement $L$. 
\end{definition}
Note that the lines $\ell_0,\ell_1,\dots,\ell_n$ of $L$ are not assumed to be in a general position. A planar rooted line arrangement may contain a subarrangement of three, or more, lines intersecting at a single point. As an immediate observation, note that by condition \eqref{PRLDef2} of the above definition, an orientation of the root line $\ell_0$ of $L$ induces a canonical orientation on each line $\ell_i$ for $i\geq 1$. Specifically, for $i\geq 1$, a directional vector of $\ell_i$ stemming out of $L_i$ is chosen to lie in the upper-half plane with respect to $\ell_0$.
\begin{figure}[H]
\resizebox{0.7\textwidth}{!}{\input{line_arr}}
\caption  
{A planar rooted line arrangement of rank $4$.}
\end{figure}
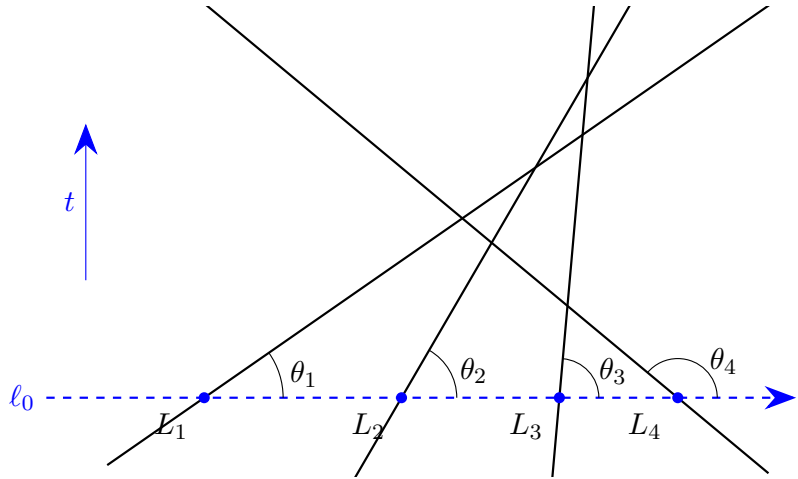 
In terms of scattering diagrams, the root line $\ell_0$ represents the one-dimensional space at time $t=0$. 
The remaining $n$ lines represent the worldline asymptotes of $n$ particles undergoing a scattering process in transition presentation. 
For $n\geq 1$, the set of all rank-$n$ planar rooted line arrangements is a real $2n$-dimensional manifold. As a possible set of global coordinates one can take positions of $L_i$ on $\ell_0$ and the oriented angles $\theta_i$ from $\ell_0$ to $\ell_i$ for $1\leq i\leq n$, or momenta $p_i=\text{ctg}(\theta_i)$ identifying this $2n$-dimensional manifold with with $UConf_n(\R)\times UConf_n((0, \pi))$,
or $UConf_n(\R)\times UConf_n(\R)$ respectively.
Here, as before, $UConf_n(X)$ is the configuration space of $n$ unlabeled, but naturally ordered in our case, points in $X$.
Note that in the copy of $UConf(\R)$ for the momenta the order of point enumeration is reversed -- the slower particles are further right. Equivalently, we may think that to each particle we assign its negative momentum.

In what follows, we will consider two examples of a non-$\Sigma$ operad structure on planar rooted line arrangements. We start off with a simpler one, but it is the second one -- affine-geometric, that will be of particular interest to us.
\subsection{Construction I -- decoupling positions and momenta.}
As we have seen in the previous section, configuration spaces $UConf_n(\R)$ assemble into an operad $\widehat{C}$, where $UConf_n(\R)=\widehat{C}(n-1)$ for $n\geq 2$.
Then the coordinate presentation $UConf_n(\R)\times UConf_n(\R)$ of rank-$n$ arrangements in terms of positions and negative momenta yields a non-$\Sigma$ operad structure on the space of all planar rooted line arrangements isomorphic to $\widehat{\C}\times\widehat{\C}$.
Again, note that a rank-$n$ arrangement is identified with an element of $\widehat{\C}(n-1)\times\widehat{C}(n-1)$. Another model thereof is an operad of polygonal chains.

For a rank-$(n+1)$ planar rooted line arrangement, or equivalently, for an initial configuration of $n+1$ particles, consider the polygonal chain $P_1P_2\dots P_{n+1}$ in $\bbR^2$, where $P_i=(q_i,-p_i)$, the initial position and the negative momentum of the $i$-th particle.
Note that both the $q$ and the $-p_i$ coordinates are strictly monotonically increasing for $1\leq i\leq n$. For all $n\geq 2$, let us define $\P(n)$ to be the set of all polygonal chains in the $qp$-plane with this property with $n+1$ vertices.
For any two polygonal chains $P\in \P(n)$, $Q\in \P(m)$ and $1\leq i\leq n$, we define their partial composition $P\circ_i Q$ as follows. The operation amounts to replacing the $i$-th line segment $P_iP_{i+1}$ with a rescaled and translated copy of $Q$ so that $Q_1$ maps to $P_{i}$ and $Q_{m+1}$ maps to $P_{i+1}$. Here, rescaling stands for the transformation $(q,p)\mapsto \left(\frac{q(P_{i+1})-q(P_i)}{q(Q_{m+1})-q(Q_1)}\cdot q, \frac{p(P_{i+1})-p(P_i)}{p(Q_{m+1})-p(Q_1)}\cdot p\right)$,
where $q(K)$ and $p(K)$ denote the $q$,$p$-coordinates of a point $K$ respectively.
This path substitution operation turns $\P$ into a non-$\Sigma$ operad isomorphic to $\widehat{\C}\times \widehat{\C}$.
Correspondingly, this operad is generated by two topological families of arity-two generators subject to quadratic relations. The latter are translated and rescaled barycentric relations \eqref{Bary}.
\begin{center}
  \begin{figure}[H]
  \resizebox{0.7\textwidth}{!}{\input{comp2}}
  \caption  
  {Evaluating $P \circ_3 Q\in\P(6)$.}
  \end{figure}
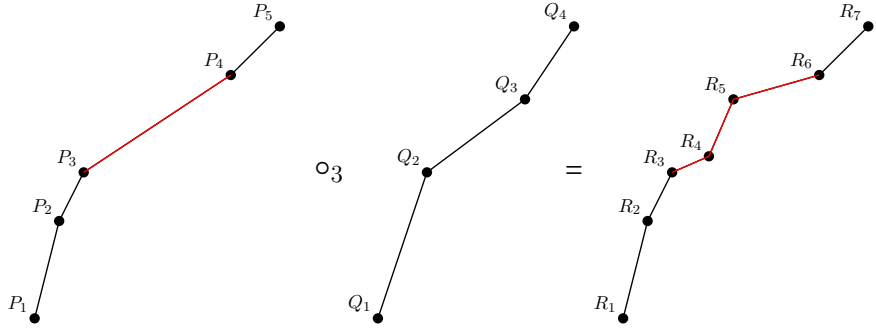
  \end{center}

\subsection{Construction II -- using affine symmetries.}
For $n\geq 1$, let $\widehat \L(n)$ be the space of all rooted line arrangements of rank $n+1$.
Let $m, n\geq 1$, $P\in \widehat \L(m)$, $Q\in \widehat \L(n)$ and~${1\leq i\leq m}$.
Consider the points of intersection of the following pairs of lines of $P$: $(p_0, p_i)$, $(p_i,p_{i+1})$, $(p_0, p_{i+1})$, and denote them by $P_i$, $P_{i,i+1}$, $P_{i+1}$ respectively. Similarly, let  $Q_1$, $Q_{1,m+1}$, $Q_{m+1}$ denote, in their respective order, the points of intersection of the following pairs of lines: $(q_0, q_1)$, $(q_1,q_{m+1})$, $(q_0, q_{m+1})$.
Here, as per our usual notation, $p_0$ and $q_0$ denote the root lines of $P$ and $Q$ respectively.
Note that neither of these two triples of points is collinear.
Then, by the fundamental theorem of affine geometry, there exists a unique orientation-preserving affine transformation 
$T_i\in Aff(\R^2)$ that sends the ordered triple ${(Q_1,Q_{1,m+1},Q_{m+1})}$ to ${(P_i,P_{i,i+1},P_{i+1})}$. Such an affine transformation sends the lines $q_0, q_1$ and $q_{m+1}$ to $p_0$, $p_i$ and $p_{i+1}$ respectively.
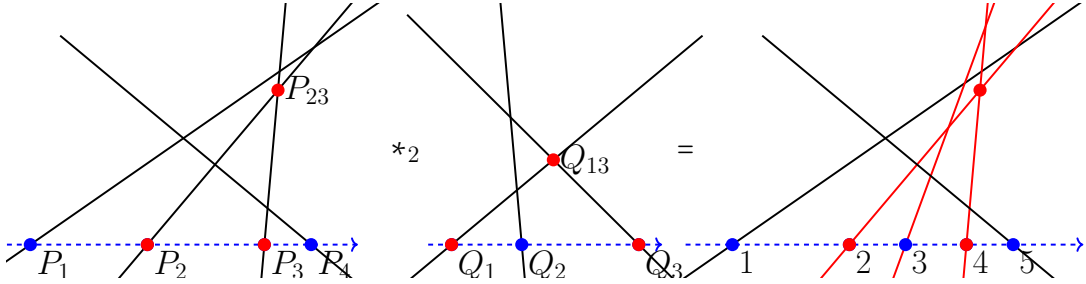
\begin{figure}[H]
\centering
\resizebox{0.9\textwidth}{!}{\input{comp1}}
\caption  
{Evaluating $P *_2 Q\in \widehat \L(4)$ for $P\in \widehat \L(3)$, $Q\in \widehat \L(2)$.}
\end{figure}
\begin{lemma}
\label{LinArrLemma}
Let $P\cup T_i(Q)$ be the union of the line arrangement $P$ and the lines in the image of the arrangement $Q$ under $T_i$. 
Then, upon appropriate enumeration of the lines, $P*_i Q:=P\cup T_i(Q)$ is a well-defined planar rooted line arrangement of rank $n+m-1$.
\end{lemma}
\begin{proof}
First, since an affine transformation $T_i$ is bijective and preserves the property of lines being parallel, no two lines in the image $T_i(Q)$ are overlapping or parallel to each other. Furthermore, the angles that the lines in the image $T_i(Q)$ make with the root line of arrangement $P$ are strictly in between the angle $\theta_i$ between the line $p_i$ and $p_0$ and the angle $\theta_{i+1}$ between the line $p_{i+1}$ and $p_0$ in arrangement $P$. Hence, no line in the image $T_i(Q)$ can be parallel to any of the lines of $P$. This verifies condition (1) of Definition \ref{PRLDef}.

By the same set of inequalities for the angles, any line in the image $T_i(Q)$ intersects any line of $P$ in the upper half-plane. Since $T_i$ respects the orientations of the root lines $p_0$ and $q_0$, all the intersections between the lines in $T_i(Q)$ lie in the upper-half plane as well. This verifies condition (2) of Definition \ref{PRLDef}.

To make condition (3) hold, we enumerate the lines in $P\cup T_i(Q)=\{r_1,r_2,\dots,r_{n+m-1}\}$ by setting
${r_j:=
\begin{cases}
p_j,\quad & j \leq i\\
T(q_{j-i+1}),\quad & i < j < i+m\\
p_{j-m+1},\quad & i+m \leq j
\end{cases}}$.
\end{proof}
We will verify that the partial compositions $(-*_i-)$ satisfy the generalized associativity conditions. 
\begin{theorem}
\label{CompThm}
The collection $\widehat\L:=\{\widehat\L(n)\}_{n\geq 1}$ with partial compositions $P*_i Q$ defined above is a non-$\Sigma$ operad in $\Top$.
\end{theorem}
\noindent
We will prove the statement for \emph{normalized} planar rooted line arrangements $\L$ first, and will extend it to $\widehat \L$ by means of proposition \ref{MovingFrameLemma}. 
The operations $(-*_i-)$ defined on $\widehat\L$ above will be identified as the partial compositions of this extended operad.
To this end, for $n\geq 1$, let $\L(n)$ be the set of all planar rooted line arrangements $P$ of rank $n+1$ bounded by the lines $t=0$, $q=0$ and $q+t=1$ in the $qt$-plane.
These are the root line $p_0$, the first line $p_1$ and the last line $p_{n+1}$ of any $P\in \L(n)$ respectively. That is, $\L$ consists of the scattering diagrams drawn in the rest frame of the first particle, with units chosen in such a way that the last particle approaches the first one at the unit speed from the right.
\begin{lemma}
\label{NormLemma}
Let $\L:=\{\L(n)\}_{n\geq 1}$ and $(-\circ_i-)$ denote the restriction of the operation $(-*_i-)$ defined above onto $\L$. Then $\L$ with these partial compositions is a non-$\Sigma$ operad.
\end{lemma}
\begin{proof}
We will identify $\L$ as a suboperad of a certain word operad $\W G$, similarly to how we did it for static particle configurations.
The relevant group for us is the four-dimensional abelian subgroup $G$ of $Aff(\R^2)$ that can be identified with the group of matrices of the form 
$
T=
\left[
\begin{array}{ccc}
  *  & *  & * \\
   0  &  * & 0 \\
   0  & 0  & 1
\end{array}
\right]$, where all the entries, except possibly $T_{1,2}$, are non-negative, and $T_{1,1}>0$.

Consider a fixed template -- the planar rooted line arrangement of rank $2$ consisting of the lines $t=0$, $q=0$, $q+t=1$
that we denote by $\tau_0$, $\tau_L$ and $\tau_R$ respectively.
Then any planar rooted line arrangement $P\in \L(n)$ can be uniquely encoded by the tuple $(T_1,\dots,T_n)$, where $T_i$ is the affine map sending points $(0,0)$, $(1,0)$, $(0,1)$ to $P_i$, $P_{i+1}$ and $P_{i,i+1}$ respectively; we use the notation as in lemma \ref{LinArrLemma}.
That is, $T_i$ keeps track on the point on intersection of the lines $p_i$ and $p_{i+1}$.
\begin{center}
\begin{figure}[H]
\resizebox{0.7\textwidth}{!}{\input{comp3}}
\end{figure}
\end{center}
Concretely, 
$T_i=
\left[
\begin{array}{ccc}
   a_{i+1}-a_i  & b_i-a_i  & a_i \\
   0  &  c_i & 0 \\
   0  & 0  & 1
\end{array}
\right]$ for $P_i=(a_i,0)$, $P_{i+1}=(a_{i+1},0)$, $P_{i,i+1}=(b_i,c_i)$.
Note that the $T_i$'s are dependent. Namely, any tuple $(T_1,\dots,T_n)$ encoding a planar rooted line arrangement is subject to the matching and normalization conditions akin to \eqref{MatchCond}, \eqref{NormCond}.
The matching condition reads
\begin{align}
\label{MatchCond2}
    T_i \tau_R = T_{i+1}\tau_L,\quad 1\leq i<n. 
\end{align}
Equivalently, it can be restated as collinearity of points
$(a_{i+1},0), (b_i,c_i), (b_{i+1},c_{i+1})$,
where the last two points can be the same.
The normalization condition amounts to 
\begin{align}
\label{NormCond2}
T_1\tau_L=\tau_L, T_n\tau_R=\tau_R.
\end{align}
Preservation of the root line, $T_i\tau_0=\tau_0$, is automatic by the form of $T_i$.
Such a tuple $(T_1,\dots,T_n)$ is an element of the word operad $\W G$, where $G$ is as defined above. The geometric definition of a composition $P*_i Q$ of planar rooted line arrangements translates to
\begin{align}
\label{TupleComp}
(T_1,\dots,T_n)\circ_i (S_1,\dots, S_m)=
(T_1,\dots, T_{i-1},T_iS_1,\dots, T_i S_m, T_{i+1},\dots, T_n)
\end{align}
in terms of tuples in $\W G$. It remains to show that the set of all tuples subject to the matching and normalization conditions stated above is closed in $\W G$.

The matching conditions for the subranges $T_1,\dots,T_{i-1}$ and $T_{i+1},\dots,T_n$ on the right-hand side of \eqref{TupleComp} follow immediately from the matching conditions for $(T_1,\dots,T_n)$. The matching condition for $T_i S_1,\dots, T_i S_m$ follows upon multiplying sides of $S_j \tau_R = S_{j+1}\tau_L$ for $1\leq j<m$
by $T_i$ on the left. In the remaining two edge cases, by matching of $T_i$'s and the normalization condition for $S_1,\dots,S_n$ we have
\[
T_{i-1}\tau_R=T_{i}\tau_L=T_{i}(S_1\tau_L)=(T_iS_1)\tau_L
\]
and similarly for the adjacent entries $T_i S_m, T_{i+1}$.
The normalization condition on the right-hand side of \eqref{TupleComp} follows directly from the one holding for $T_1,T_n$ and $S_1,S_m$.
\end{proof}
\begin{proof}[Proof (of theorem \ref{CompThm}).]
The argument is analogous to the one used for non-$\Sigma$ operad $\widehat C$ of static particle configurations discussed in the preceding section.
First, note that any planar rooted line arrangement $P\in \widehat \L(n)$ is determined by a tuple $(T_1,\dots,T_n)\in \W G(n)$ subject to the matching condition \eqref{MatchCond2}, but with no regard for the normalization \eqref{NormCond2}. The tuples of this form are not closed under the partial compositions of the word operad $\W G$. A correction is provided by proposition \ref{MovingFrameLemma} upon relating $\widehat \L$ to non-$\Sigma$ operad $\L$ of normalized planar rooted line arrangements defined above.

Indeed, since \eqref{MatchCond2} remains valid upon multiplying both sides by any $g\in G$ on the left, $\widehat \L(n)$ is naturally acted on by $G$ for any $n\geq 1$. The moving frame $\rho_n: \widehat \L(n)\to G$ that normalizes $\widehat \L$ onto $\L$ is defined by assigning to $P\in \widehat\L(n)$ the unique affine transform that sends points $(0,0)$, $(1,0)$, $(0,1)$ to $P_1$,$P_{n+1}$ and $P_{1,n+1}$ respectively. Note that it always exists and can be worked out explicitly as a $3$-by-$3$ matrix, even though the resulting formula is cumbersome. Now, an application of proposition \ref{MovingFrameLemma} extends the non-$\Sigma$ operad structure provided by lemma \ref{NormLemma} from $\L$ to $\widehat \L$. The geometric description of $P*_i Q$ stated before matches \eqref{MovingFrameComp} provided by proposition \ref{MovingFrameLemma}.
\end{proof}
\begin{corollary}
 Let $F$ be any planar rooted line arrangement. Then all planar rooted line arrangements containing $F$ as a subarrangement form a suboperad of $\hL$. Indeed, since a partial composition $P*_i Q$ amounts to only drawing new lines on $P$, then $F\subset P$ persists in $P*_i Q$.
\end{corollary}

By identifying a rank-$n$ planar rooted line arrangement with the data of a purely elastic scattering process of $n$ particles, the construction of the $\hL$ yields an example of an operad structure on initial particle configurations that we have been after. 
By theorem \ref{CompThm}, to compute a composition of initial particle configurations $a\circ_i b$, we switch to the normalized rest frame of the first particle $a_1$, do the work there according to the affine-geometric rules prescribed by operad $\L$, and return the answer back covariantly in the reference frame of $a$. This is what $\hL$ does. Naturally, thanks to the natural $G$-action on $\hL$, the result can be returned in the reference frame of $b$, or elsewhere.
\begin{lemma}
\label{YBLemma}
Non-$\Sigma$ operad $\widehat \L$ is generated by the symmetry group $\widehat \L(1)=G$ and three topological families of binary operations corresponding to three possible combinatorial types of rank-3 arrangements, shown below.
\end{lemma}
\begin{center}
  \begin{figure}[H]  
  \resizebox{0.8\textwidth}{!}{\input{yb}}
  \caption  
  {Combinatorial types of rank-$3$ arrangements.}
  \label{ybconf}
  \end{figure}
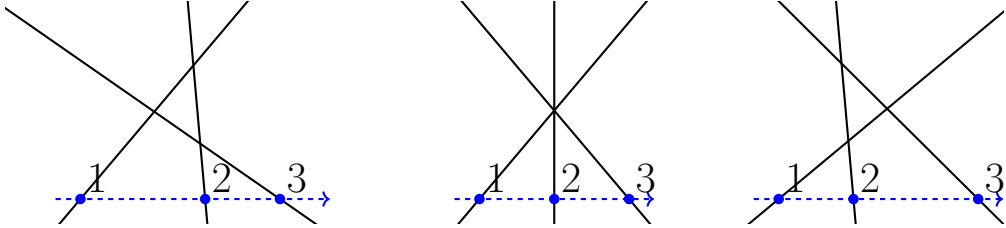
  \end{center}
\begin{proof}
The proof amounts to observing that evaluating a partial composition $Q *_i M$ with either one of the line arrangements $M$ on figure \ref{ybconf} amounts to drawing a new line in between the lines $q_i$ and $q_{i+1}$ of $Q$. Any diagram with $n>2$ lines can be produced in this way by starting with a planar root arrangement of rank $2$ and consecutively adding more lines one by one.
More formally, let $P\in \widehat \L(n)$ for $n>2$ and ${P'=\{p_0,p_1,\dots,p_{n-1},p_{n+1}\}}$ be the arrangement obtained by removing the penultimate line $p_n$ of $P$. Then $P'\in\hL(n-1)$, and $P=P'*_{n-1} M$, where $M=\{p_0,p_{n-1},p_n,p_{n+1}\}$. The latter is necessarily of one of the above forms, and induction applied to $P'$ yields the result.
\end{proof}

\section{Symmetries of $\hL$ and factorized scattering}
Let us address the question of what symmetries can non-$\Sigma$ operad of scattering diagrams $\hL$ support. 
In what follows, $G$ denotes the subgroup of $Aff(\R)$ described in lemma \ref{NormLemma}.
First, there is a global gauge action by adjunction $g\hL g^{-1}$, where $g\in Aff(\R^2)$ acts on each $\hL(n)\in G^n\leq Aff(\R^2)$ diagonally; cf. lemma~\ref{TilingOpProp}.
Also, partial compositions commute with translations along the $q$-axis, or for that matter, with multiplication action by $g\in G$.
In our non-relativistic setting the crossing symmetry does not arise.
As far as the CPT goes, no charge conjugation is defined, since the lines of a planar rooted arrangements as such are not assumed to have any internal degrees of freedom. 
That said, one may consider a semidirect product $\hL \rtimes H$, where $H$ is the symmetry group of a particle. Note that in this setting the charges would be defined up to a global phase. Indeed, an element of $(\hL \rtimes H)(n)$ is a rank-$(n+1)$ line arrangement with $n$ values of $H$ distributed among $n+1$ lines (particles).

The $P$-symmetry amounts to switching the direction of the root line, or equivalently, to flipping a diagram about the $t$-axis. 
The equivariance with respect to this involution $\pi$ reads $\pi(P *_i Q)=\pi(P)*_{m-i+1}\pi(Q)$ for any $P\in\hL(m)$, $Q\in \hL(n)$, $1\leq i\leq m$. In fact, such a symmetry exists in any word operad $\W M$, where $\pi$ reverses a word.

The $T$-symmetry can be introduced upon extending our definition of a planar rooted line arrangements to configurations where line intersections may appear in both the upper and the lower half-planes (but not on the root line). The symmetry group $G$ should be modified appropriately; the construction of $\hL$ remains the same. The time reversal $\tau$ amounts, in such a setting, to flipping a diagram vertically about the $q$-axis. The equivariance reads $\tau(P\circ_i Q)=\tau(P)\circ_{m-i+1}\tau(Q)$ for $P$,$Q$ as above. 

A crucial property for integrability of both classical and quantum pure elastic scattering processes are good covariance properties of scattering data with respect to  \emph{Baxter's $Z$-symmetry}. By definition, $Z$-symmetry on $P\in \hL(n)$ amounts to shifting any line of $P$ parallel to itself without violating the ordered momenta condition. 
For instance, different combinatorial types of rank-$3$ arrangements, as shown on figure \ref{ybconf}, can be related to each other by such a transformation. More generally, using $Z$-symmetry, any rank-$(n+1)$ arrangement $P\in \hL(n)$ can be transformed to an arrangement of $n+1$ lines intersecting at a single point. We can describe all of them in the following way. Let $\ell_0$ be a root line and $A$ be a point in the upper half-plane. Since affine maps preserve intersections, the space $\widehat{L}_A$ of all planar rooted line arrangements consisting of lines intersecting at $A$ is a suboperad of $\widehat L$. Moreover, since affine maps preserve distance ratios along a single line, $\widehat{L}_A$ is isomorphic to the operad of static particle configurations $\widehat\C$. Indeed, all the momenta can be recovered from the particle positions and $A$.
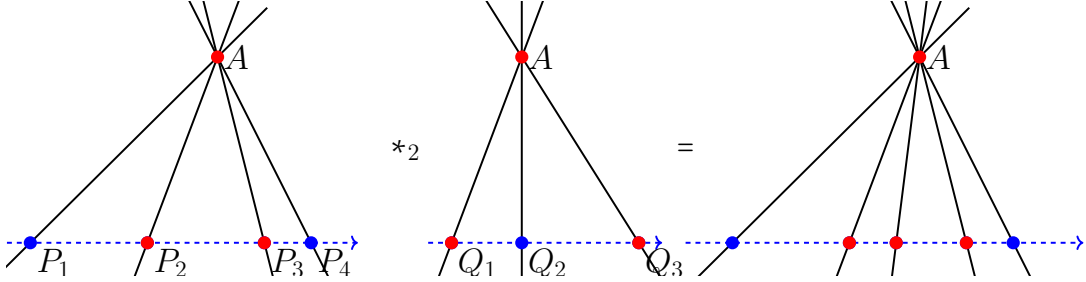
\begin{figure}[H]
\centering
\resizebox{0.9\textwidth}{!}{\input{comp4}}
\caption  
{Evaluating $P *_2 Q\in \widehat \L_A(4)$ for $P\in \widehat \L_A(3)$, $Q\in \widehat \L_A(2)$.}
\end{figure}
Since parallel translations commute with affine maps, there is a surjective non-$\Sigma$ operad morphism ${p:\hL\to \hL_A\simeq \widehat \C}$ for any $A$ in the upper half-plane. This is the $Z$-symmetry transformation of $P$ into a configuration of lines intersecting at $A$ described above.
These symmetries are useful if we would like to study invariants of planar rooted line arrangements. In terms of our physical model, a quest for invariants can be formalized as follows.
\begin{definition}
 Let $C$ be a symmetric monoidal category subject to our usual technical constraints, and $\P$ be a non-$\Sigma$ operad in $C$. A \emph{$C$-valued} (1+1 dimensional purely elastic) \emph{scattering theory} is a non-$\Sigma$ operad morphism $s:\hL \to \P$ over $\Sets$. 
\end{definition}
By lemma \ref{YBLemma}, defining $s$ amounts to specifying $s$ on all planar rooted line arrangements of ranks $2$ and $3$. 
Line arrangements of rank $2$ is the monoid $\hL(1)$, which is identified with our affine symmetry group ${G\leq Aff(\R^2)}$. So $\P(1)$ contains a homomorphic image thereof.
To handle rank-$3$ arrangements, we can impose, for the sake of simplicity, the $Z$-symmetry and consider scattering theories that do not distinguish between $Z$-symmetric scattering diagrams of different combinatorial types shown on figure \ref{ybconf}. These are scattering theories where multiparticle collisions can be reduced to pairwise interactions.
Specifically, we will say that $s$ has the \emph{Yang-Baxter property}
if $s$ factorizes through $p$:
 \begin{center}
 \begin{tikzpicture}[node distance=2.5cm, auto, scale=0.5]
  \node (L) at (0,0) {$\hL$};
  \node (C) at (0,-3) {$\hL_A$};
  \node (P) at (3,0) {$\P$};

  \draw[->] (L) -- (P) node[midway, above] {$s$};
  \draw[->>] (L) -- (C) node[midway, left] {$p$};
  \draw[->, dashed] (C) -- (P) node[midway, above right] {};
\end{tikzpicture}
\end{center}
By the global gauge action on $\hL$ this is independent from the choice of $A$. As a concrete implementation, a $\kVect$-valued scattering theory in the above sense can be constructed from the data of an $R$-matrix, or equivalently, by recasting Zamolodchikov-Faddeev's algebras in operadic terms.

\section{Outlook}
A merit of the operadic approach to the subject is in its functoriality.
This becomes relevant if we are interested in generalizations and deformations of scattering theories.
In particular, it could be interesting to look for scattering theories $\hL \to \P$, where $\P$ is valued in a category with good homotopy properties, as well as to modify the source operad $\hL$, for instance, by considering pseudoline arrangements with some additional data that would yield a reasonable operad structure thereon. 
One may look for scattering theories associated to various combinatorial and algebraic invariants of line arrangements or configuration spaces of points on a line. We briefly elaborate on these points below.
\subsection{Related polytopes.}
A construction that we may naturally associate to a planar rooted line arrangement is a polytopal chain in a permutahedron. Indeed, in transmission presentation, a scattering process of $n$ particles can be thought of as a path in $\R^n$ with the braid hyperplane arrangement.
A permutahedron is a dual object thereof.

More specifically, for any $t>0$, the cross section of a scattering diagram by a horizontal line $t=const$ is, generically, a configuration of $n$ distinct points in $\R$, but there may (and will, for $n>1$) be up to $\frac{n(n-1)}{2}$ special values of $t$ corresponding to the time moments when a lines intersection, or a multiple lines intersections, take place for some $t$. Any collision pattern at moment $t$ is described by a partition $I(t)=\{I_1(t),I_2(t),\dots,I_k(t)\}$ of the set $\{1,2,\dots,n\}$, where a block $I_j(t)$ of cardinality $|I_j(t)|>2$ registers the fact of a simultaneous collision of particles $\alpha\in I_j(t)$ at moment $t$, or equivalently, represents a point of common intersection of the lines $\ell_\alpha$ for $\alpha\in I_j(t)$. In particular, generically, in absence of collisions, $I(t)$ is the finest partition $\{\{1\},\{2\},\dots,\{n\}\}$. The simultaneous collision of all $n$ particles is represented by the trivial single block partition.
In a general position, any planar rooted lines arrangement has $\frac{n(n-1)}{2}$ singular values of $t$, when $I(t)$ is a partition with $n-2$ blocks of size $1$ and one extra block $\{i, j\}$.

Furthermore, for any $t>0$, the partition $I(t)$ carries an extra structure, and in fact can be equipped with a natural total order on the set of blocks. Indeed, invoking the orientation on the ambient $\bbR^2$, we can define $I'<I''$ for $I',I''\in I(t)$ iff on the horizontal cross section of a scattering diagram for a given $t$ the intersection of lines indexed by $I'$ appears strictly to the left of the intersection of the lines indexed by $I''$. Now, the set of all ordered partitions of $\{1,2,\dots,n\}$ is a lattice that admits a polytopal realization as the face lattice of the $(n-1)$-dimensional permutahedron $P_{n-1}$. Indeed, this is one of the known compactifications of the configuration space of points in a closed interval \cite{Lambrechts}. Thus a scattering diagram gives rise to a polytopal chain in $P_{n-1}$ that starts at the vertex $(\{1\},\{2\},\dots,\{n\})$ and arrives at the vertex 
$(\{n\},\{n-1\},\dots,\{1\})$. 
For a rank $n$ planar rooted line arrangement in a general position, such a chain will traverse only the edges and the vertices of $P_{n-1}$. This corresponds to having strictly pairwise collisions. 
As an example, the three combinatorially distinct rank $3$ planar rooted line arrangements shown on figure  \ref{ybconf} give rise to three paths on a hexagon $P_2$ as shown on figure \ref{hex}.
\begin{center}
  \begin{figure}[H]  
  \resizebox{0.4\textwidth}{!}{\input{hex}}
  \caption  
  {Three chains in $P_2$.}
  \label{hex}
  \end{figure}
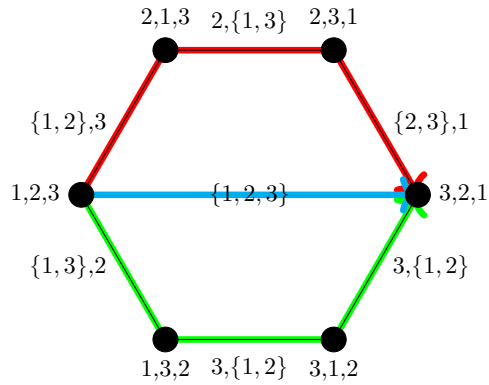
  \end{center}
The two sides of the Yang-Baxter identity $S_{12}S_{13}S_{23}=S_{23}S_{13}S_{12}$ are encoded by two paths, consisting of three edges each, connecting the vertices $\{1\}\{2\}\{3\}$ and $\{3\}\{2\}\{1\}$.  The scattering scenarios with only pairwise interactions are parametrized by reduced decompositions of the longest element $n\,(n-1)\dots 1$ of $\S_n$.
\subsection{Relation to the Swiss-cheese operad}
Given a line arrangement $P\in \hL(n)$, we define its \textit{upper envelope} $U_P$ as the unique polygonal chain made of the segments lying on the lines of $P$ such that all points of intersection of the lines of $P$ lie in the closed region bounded by $U_P$ and the root line. That is, constructing $U_P$ amounts to going along the lines of $P$, starting at point $P_1$ and finishing at $P_{n+1}$, in such a way that all points of intersection of the lines of $P$ remain on the right hand side with respect to the direction of motion.

For a given $P\in \hL(n)$, there exists a sufficiently small $\epsilon>0$ such that thickening each line segments of $U_P$ to a rectangular band of width $\epsilon$ results in a subset $P_\epsilon$ of $\mathbb{R}^2$ homeomorphic [this needs to be refined] to a half disk with $n$ half-disks at the bottom cut out and with $m$ holes in the bulk, where $m$ is the number of the connected bounded subsets of $\bbR^2\setminus P$. In particular, $m=\frac{n(n-1)}{2}$ if $P$ is in a general position.
\begin{center}
  \begin{figure}[H]
  \resizebox{0.7\textwidth}{!}{\input{strips}}
  \caption  
  {Morphing a planar rooted line arrangement of rank $4$ to an element of $\mathcal{SC}(3,3)$.}
  \end{figure}
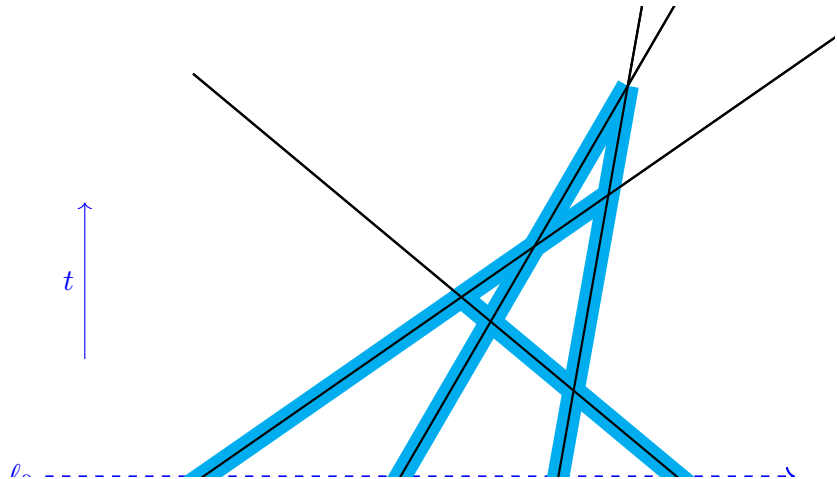
  \end{center}
That is, $P\in \hL(n)$ gives rise (non-uniquely) to an element $P_\epsilon\in \mathcal{SC}(n,m)$ of the Swiss-cheese operad $\mathcal{SC}$\cite{Voronov}.

\label{CSWPSec}


\subsection*{Acknowledgments}
The author is grateful to Branislav Jur\v{c}o, Tom\'{a}\v{s} Proch\'{a}zka, Dmitry Roytenberg and Alexander Voronov for stimulating communications.
The work is supported by RVO:67985840 of the Institute of Mathematics of the Czech Academy of Sciences and the Praemium Academiae grant of M. Markl. 

\printbibliography

\end{document}

%% file: sct1.tex
\begin{tikzpicture}[>={stealth[length=3mm]}]

\def\ball#1#2#3#4#5{
    \draw[->, thick] (#5)  -- ++(#3,0); 
    \draw[fill=#1] (#5) circle (#2); 
    \node at (#5) {#4}; 
}

\coordinate (A1) at (0,0);
\coordinate (A2) at (2,0);
\coordinate (A3) at (3.5,0);
\coordinate (A4) at (6,0);

\coordinate (B1) at (6,4);
\coordinate (B2) at (5,4);
\coordinate (B3) at (2,4);
\coordinate (B4) at (0,4);

\draw[gray, very thick, ->] (-1,-0.18) -- (7,-0.18);

\foreach \i in {1,2,3,4} 
{
\draw ($(A\i)!-.12!(B\i)$) -- ($(B\i)!-0!(A\i)$);
}

\ball{cyan!40}{0.15}{1}{}{A1};
\ball{cyan!40}{0.15}{0.7}{}{A2};
\ball{cyan!40}{0.15}{-0.5}{}{A3};
\ball{cyan!40}{0.15}{-1}{}{A4};


\draw [->] (-1,1) -- node[left] {$t$} (-1,2);

\end{tikzpicture}

%% file: sct2.tex
\begin{tikzpicture}[>={stealth[length=3mm]}]

\def\ball#1#2#3#4#5#6{%
    \draw[->, thick] (#5 + #2, #6 + #2) -- (#5 + #3, #6 + #2);
    \draw[fill=#1] (#5, #6 + #2) circle (#2);
    \node at (#5, #6 + #2) {#4};
    
}

\draw[draw=red, dashed, fill=red!10] (1.8, 1.2) -- (2, 0) -- (4.3, 0) -- (4.8, 1.2) -- cycle;

\draw[gray, very thick] (-1,0) -- (7,0);

\ball{cyan!40}{0.2}{1}{}{0}{0}
\ball{cyan!40}{0.2}{0.7}{}{2}{0}
\ball{cyan!40}{0.2}{0.5}{}{4.3}{0}
\ball{cyan!40}{0.2}{-0.5}{}{6}{0}

\draw[gray, very thick] (1,1.2) -- (5.6,1.2);

\ball{green!40}{0.2}{0.6}{}{1.8}{1.2}
\ball{green!40}{0.2}{0.5}{}{3.2}{1.2}
\ball{green!40}{0.2}{0.4}{}{4.8}{1.2}

\node [draw=none, rotate=-90] at (3.2, 0.6) {\Large$\Rightarrow$};

\node [draw=none, rotate=-90] at (3.2, -0.6) {\Large$=$};

\node at (0.8, 1.2) {$b$};
\node at (-1.2, 0) {$a$};

\node at (1, -0.2) {$1$};
\node at (3.15, -0.2) {$2$};
\node at (5.15, -0.2) {$3$};

\draw[gray, very thick] (-1,-1.4) -- (7,-1.4);

\ball{teal}{0.2}{1}{}{0}{-1.4}
\ball{teal}{0.2}{0.7}{}{2}{-1.4}
\ball{teal}{0.2}{0.6}{}{3.2}{-1.4}
\ball{teal}{0.2}{0.5}{}{4.3}{-1.4}
\ball{teal}{0.2}{-0.5}{}{6}{-1.4}

\node at (-1.5, -1.4) {$a\circ_2 b$};

\end{tikzpicture}

%% file: sct3.tex
\begin{tikzpicture}[>={stealth[length=3mm]}]

\def\ball#1#2#3#4#5#6{%
    \draw[fill=#1] (#5, #6 + #2) circle (#2);
    \node at (#5, #6 + #2) {#4};
    \draw[->, thick] (#5 + #2, #6 + #2) -- (#5 + #3, #6 + #2);
}


\draw[draw=red, dashed, fill=red!10] (1.4, 2.4) -- (1.7, 1.2) -- (3.2, 1.2) -- (5.2, 2.4) -- cycle;

\draw[draw=red, dashed, fill=red!10] (1.8, 1.2) -- (2, 0) -- (4.6, 0) -- (4.8, 1.2) -- cycle;

\draw[gray, very thick] (-1,0) -- (7,0);

\ball{cyan!40}{0.2}{1}{}{0}{0}
\ball{cyan!40}{0.2}{0.7}{}{2}{0}
\ball{cyan!40}{0.2}{0.5}{}{4.6}{0}
\ball{cyan!40}{0.2}{0.4}{}{6}{0}

\draw[gray, very thick] (1,1.2) -- (5.6,1.2);

\ball{green!40}{0.2}{0.6}{}{1.8}{1.2}
\ball{green!40}{0.2}{0.5}{}{3.2}{1.2}
\ball{green!40}{0.2}{0.4}{}{4.8}{1.2}

\node [draw=none, rotate=-90] at (3.2, 0.6) {\Large$\Rightarrow$};

\draw[gray, very thick] (1,2.4) -- (5.6,2.4);

\ball{orange!40}{0.2}{1}{}{1.4}{2.4}
\ball{orange!40}{0.2}{0.5}{}{4}{2.4}
\ball{orange!40}{0.2}{0.4}{}{5.2}{2.4}

\draw[gray, very thick] (-1,-1.2) -- (7,-1.2);

\ball{teal}{0.2}{1}{}{0}{-1.2}
\ball{teal}{0.2}{0.7}{}{2}{-1.2}
\ball{teal}{0.2}{0.55}{}{3}{-1.2}
\ball{teal}{0.2}{0.52}{}{3.8}{-1.2}
\ball{teal}{0.2}{0.5}{}{4.6}{-1.2}
\ball{teal}{0.2}{0.4}{}{6}{-1.2}

\node [draw=none, rotate=-90] at (3.2, 0.6) {\Large$\Rightarrow$};

\node [draw=none, rotate=-90] at (2.6, 1.8) {\Large$\Rightarrow$};

\node at (-2, -1.2) {$a\circ_2(b\circ_1 c)$};
\node at (-1.2, 0) {$a$};
\node at (0.8, 1.2) {$b$};
\node at (0.8, 2.4) {$c$};


\end{tikzpicture}

%% file: comp_tree.tex
\begin{tikzpicture}[scale=1, grow=up, level distance=1cm, sibling distance=1.5cm, inner sep=2mm]

\node[draw=none] {}
child 
{
node[circle, draw, fill=cyan!40] {}
child { node {} }
child { 
node [circle, draw, fill=green!40] {}
child {}
child {
node [circle, draw, fill=orange!40] {}
child{}
child{}
}
}
child { node {} }
};

\end{tikzpicture}

%% file: grafting.tex
\begin{tikzpicture}[scale=1, grow=up, level distance=1cm, sibling distance=1.5cm, inner sep=2mm]

    \node[draw=none] (A0) {}
    child {node[circle, draw] (A1) {}
        child { node(A5) {\Large 4} }
        child { node[circle, draw] (A2) {}
            child { node (A3) {\Large 3} }
            child { node (A4) {\Large 2} }
        }
        child { node (A6) {\Large 1} }        
    };

    \node at (2.2, 1) {\Huge $\circ_4$};
    
    \node[draw=none, right=4cm of A0] (B0) {}
    child {node [draw, circle] (B1) {} 
        child { node (B2) {\Large 3} }
        child { node (B3) {\Large 2} }
        child { node (B4) {\Large 1} }};
  
    \node at (7.3, 1) {\Huge =};
    
    \node[draw=none, right=6cm of B0] (C0) {}
        [sibling distance=2.6cm]   
        child 
        {node[circle, draw] (C1) {}
        child { node[circle, draw] (C2) {}
            [sibling distance=1.5cm]
            child { node (C3) {\Large 6} }
            child { node (C4) {\Large 5} }
            child { node (C5) {\Large 4} }
        }
        child { node[draw, circle] (C6) {}
            [sibling distance=1.5cm]
            child { node (C7) {\Large 3} }
            child { node (C8) {\Large 2} }
        }
        child { node(C9) {\Large 1} }};

   \end{tikzpicture}

%% file: concat.tex
\begin{tikzpicture}[scale=1, grow=right,level distance=1cm, sibling distance=1.5cm, inner sep=2mm]

    \node[draw=none] (A0) {}
    child {node[circle, draw] (A1) {}
        child {node[circle, draw] (A2) {}                            
                    child {node[circle, draw] (A3) {}
                        child {node[draw=none] (A4) {}}}}};

    \draw[-{Latex[length=3mm]}] (A1) -- (A0);
    \node[right of=A4] (circ) {\Huge $\circ_1$};

    \node[draw=none, right of=circ] (B0) {}
    child {node[circle, draw] (B1) {}
        child {node[circle, draw] (B2) {}
            child {node[draw=none] (B3) {}}}};
    \draw[-{Latex[length=3mm]}] (B1) -- (B0);
    \node[right of=B3](eq) {\Huge =};
    
    \node[draw=none, right of=eq] (C0) {}
    child {node[circle, draw] (C1) {}
        child {node[circle, draw] (C2) {}
            child {node[circle, draw] (C3) {}
                child {node[circle, draw] (C4) {}
                    child {node[circle, draw] (C5) {}                                               
                        child {node[draw=none] (C6) {}}}}}}};
    
    \draw[-{Latex[length=3mm]}] (C1) -- (C0);
\end{tikzpicture}

%% file: ascomp.tex
\begin{tikzpicture}[scale=1, grow=up, level distance=1.2cm, sibling distance=1.5cm, inner sep=2mm]

    \node[draw=none] (A0) {}
        child
        { node[circle,draw] (A1) {}
        child { node(A4) {\Large 4} }
        child { node(A4) {\Large 3} }
        child { node(A3) {\Large 2} }
        child { node (A2) {\Large 1} }
        };
  
    \node at (2.2, 1) {\Huge $\circ_2$};
    
    \node[draw=none, right=4cm of A0] (B0) {}
        child
        { node[circle,draw] (B1) {}
        child { node (B2) {\Large 3} }
        child { node (B3) {\Large 2} }
        child { node (B4) {\Large 1} }};
  
    \node at (7, 1) {\Huge =};
    
    \node[draw=none, right=4cm of B0] (C00) {}
        child
        { node[circle,draw] (C1) {}
        child { node(C6) {\Large 6}}
        child { node(C7) {\Large 5}}
        child { node[circle, draw] (C2) {}
            [sibling distance=1.5cm]
            child { node (C3) {\Large 4} }
            child { node (C4) {\Large 3} }
            child { node (C5) {\Large 2} }        
        }
        child { node(C0) {\Large 1}}};
    \node at (12, 1) {\Huge $\equiv$};
    \node[draw=none,right=5cm of C00] (D0) {}
        child
        { node[circle,draw] (D00) {}
        [sibling distance=1cm]
        child { node(D6) {\Large 6}}
        child { node(D5) {\Large 5}}        
        child { node(D4) {\Large 4}}
        child { node(D3) {\Large 3}}
        child { node(D2) {\Large 2}}
        child { node(D1) {\Large 1}}};
    
   \end{tikzpicture}

%% file: wordop_corolla.tex
\begin{tikzpicture}[scale=1, grow=up, level distance=1cm, sibling distance=1.5cm, inner sep=2mm]
    \node[draw=none] (D0) {}
        child
        { node[circle,draw] (D00) {}
        [sibling distance=1cm]
        child { node(D5) {\Large $a_5$}}        
        child { node(D4) {\Large $a_4$}}
        child { node(D3) {\Large $a_3$}}
        child { node(D2) {\Large $a_2$}}
        child { node(D1) {\Large $a_1$}}};

   \end{tikzpicture}

%% file: d1.tex
\begin{tikzpicture}[thick, scale=2]

\begin{scope}
\draw[black] (0,0) -- (1,0);

\filldraw[black] (0,0) circle[radius=0.2pt];
\filldraw[black] (1,0) circle[radius=0.2pt]; 

\draw[cyan!40, very thick] (0.1,0) -- (0.22,0) node[midway, above, black, scale=0.3] {1};
\filldraw[blue] (0.1,0) circle[radius=0.2pt];
\filldraw[blue] (0.22,0) circle[radius=0.2pt];

\draw[cyan!40, very thick] (0.4,0) -- (0.6,0)
node[midway, above, black, scale=0.3] {2};
\filldraw[blue] (0.4,0) circle[radius=0.2pt];
\filldraw[blue] (0.6,0) circle[radius=0.2pt];

\draw[cyan!40, very thick] (0.7,0) -- (0.9,0)
node[midway, above, black, scale=0.3] {3};
\filldraw[blue] (0.7,0) circle[radius=0.2pt];
\filldraw[blue] (0.9,0) circle[radius=0.2pt];

\end{scope}

\end{tikzpicture}

%% file: lshape_tree.tex
\begin{tikzpicture}[scale=1, grow=up, level distance=1.2cm, sibling distance=2cm, inner sep=2mm,
level 3/.style={sibling distance=0.6cm}]

\node[draw=none] {}
child{
    node[draw, circle] {}    
    child {}
    child {}
    child {
        node[draw, circle] {}
        child {
            node[draw, circle] {}
            child {}
            child {}
            child {}
            child {}
        }
        child {}
        child {}
        child {}
    }
    child {
        node[draw, circle] {}
        child {}
        child {}
        child {}
        child {}
    }
};

\end{tikzpicture}

%% file: conf.tex
\begin{tikzpicture}[thick, scale=2]

\begin{scope}
\draw[black] (0,0) -- (1,0);

\draw[cyan!40, very thick] (0,0) -- (0.22,0) node[midway, above, black, scale=0.3] {1};

\draw[cyan!40, very thick] (0.22,0) -- (0.6,0)
node[midway, above, black, scale=0.3] {2};

\draw[cyan!40, very thick] (0.6,0) -- (1,0)
node[midway, above, black, scale=0.3] {3};

\filldraw[cyan] (0,0) circle[radius=0.2pt];
\filldraw[cyan] (1,0) circle[radius=0.2pt]; 
\filldraw[blue] (0.22,0) circle[radius=0.2pt];
\filldraw[blue] (0.6,0) circle[radius=0.2pt];

\end{scope}

\end{tikzpicture}

%% file: conf2.tex
\begin{tikzpicture}[thick, scale=2]

\begin{scope}
\draw[black] (0,0) -- (1,0);

\draw[cyan!40, very thick] (0,0) -- (0.22,0) node[midway, above, black, scale=0.3] {1};

\draw[cyan!40, very thick] (0.22,0) -- (0.6,0)
node[midway, above, black, scale=0.3] {2};

\draw[cyan!40, very thick] (0.6,0) -- (1,0)
node[midway, above, black, scale=0.3] {3};

\filldraw[cyan] (0,0) circle[radius=0.2pt];
\filldraw[cyan] (1,0) circle[radius=0.2pt]; 
\filldraw[blue] (0.22,0) circle[radius=0.2pt];
\filldraw[blue] (0.6,0) circle[radius=0.2pt];

\end{scope}

\begin{scope}[shift={(1.5,0)}]
\draw[black] (0,0) -- (1,0);

\draw[cyan!40, very thick] (0,0) -- (0.2,0) node[midway, above, black, scale=0.3] {1};

\draw[cyan!40, very thick] (0.2,0) -- (1,0)
node[midway, above, black, scale=0.3] {2};

\filldraw[cyan] (0,0) circle[radius=0.2pt];
\filldraw[blue] (0.2,0) circle[radius=0.2pt];
\filldraw[cyan] (1,0) circle[radius=0.2pt];
\end{scope}

\begin{scope}[shift={(0.75,-0.3)}]
\draw[black] (0,0) -- (1,0);

\draw[cyan!40, very thick] (0,0) -- (0.22,0) node[midway, above, black, scale=0.3] {1};

\draw[cyan!40, very thick] (0.22,0) -- (0.296,0)
node[midway, above, black, scale=0.3] {2};

\draw[cyan!40, very thick] (0.296,0) -- (0.6,0)
node[midway, above, black, scale=0.3] {3};

\draw[cyan!40, very thick] (0.6,0) -- (1,0)
node[midway, above, black, scale=0.3] {4};

\filldraw[cyan] (0,0) circle[radius=0.2pt];
\filldraw[cyan] (1,0) circle[radius=0.2pt]; 
\filldraw[blue] (0.22,0) circle[radius=0.2pt];
\filldraw[blue] (0.296,0) circle[radius=0.2pt];
\filldraw[blue] (0.6,0) circle[radius=0.2pt];

\end{scope}

\node[scale=0.5] at (1.25,0) {$\circ_2$};
\node[scale=0.5] at (1.25,-0.12) {$=$};

\end{tikzpicture}

%% file: conf3.tex
\begin{tikzpicture}[thick, scale=2]

\begin{scope}
\draw[black] (0,0) -- (1,0);

\draw[cyan!40, very thick] (-0.1,0) -- (0.22,0) node[midway, above, black, scale=0.3] {1};

\draw[cyan!40, very thick] (0.22,0) -- (0.6,0)
node[midway, above, black, scale=0.3] {2};

\draw[cyan!40, very thick] (0.6,0) -- (1.1,0)
node[midway, above, black, scale=0.3] {3};

\filldraw[blue] (0,0) circle[radius=0.2pt];
\filldraw[blue] (1,0) circle[radius=0.2pt]; 
\filldraw[blue] (0.22,0) circle[radius=0.2pt];
\filldraw[blue] (0.6,0) circle[radius=0.2pt];

\end{scope}

\begin{scope}[shift={(1.5,0)}]

\draw[cyan!40, very thick] (-0.1,0) -- (0.2,0) node[midway, above, black, scale=0.3] {1};

\draw[cyan!40, very thick] (0.2,0) -- (1.1,0)
node[midway, above, black, scale=0.3] {2};

\filldraw[blue] (0,0) circle[radius=0.2pt];
\filldraw[blue] (0.2,0) circle[radius=0.2pt];
\filldraw[blue] (1,0) circle[radius=0.2pt];
\end{scope}

\begin{scope}[shift={(0.75,-0.3)}]
\draw[black] (0,0) -- (1,0);

\draw[cyan!40, very thick] (-0.1,0) -- (0.22,0) node[midway, above, black, scale=0.3] {1};

\draw[cyan!40, very thick] (0.22,0) -- (0.296,0)
node[midway, above, black, scale=0.3] {2};

\draw[cyan!40, very thick] (0.296,0) -- (0.7,0)
node[midway, above, black, scale=0.3] {3};

\draw[cyan!40, very thick] (0.6,0) -- (1.1,0)
node[midway, above, black, scale=0.3] {4};

\filldraw[blue] (0,0) circle[radius=0.2pt];
\filldraw[blue] (1,0) circle[radius=0.2pt]; 
\filldraw[blue] (0.22,0) circle[radius=0.2pt];
\filldraw[blue] (0.296,0) circle[radius=0.2pt];
\filldraw[blue] (0.6,0) circle[radius=0.2pt];

\end{scope}

\node[scale=0.5] at (1.25,0) {$\circ_2$};
\node[scale=0.5] at (1.25,-0.12) {$=$};

\end{tikzpicture}

%% file: line_arr.tex
\begin{tikzpicture}[>={Stealth[length=4mm, width=3mm]}]


\coordinate (P1) at (2, 0);
\coordinate (P2) at (4.5, 0);
\coordinate (P3) at (6.5, 0);
\coordinate (P4) at (8, 0);

\clip (-1,-1) rectangle (10,5);


\draw [->,shorten >= -1.5cm, dashed, blue, thick, name path=root_line] node[left] {$\ell_0$} (0,0) -- (P4);
\draw[thick, shorten <= -1.5cm, name path=line1] (P1)  -- +(35: 10cm);
\draw[thick, shorten <= -1.5cm, name path=line2] (P2) -- +(60: 8cm);
\draw[thick, shorten <= -1.5cm, name path=line3] (P3) -- +(85: 8cm);
\draw[thick, shorten <= -1.5cm, name path=line4] (P4) -- +(140: 8cm);

\path[name path=upper_line] (0,5)--(10,5);
 
\foreach \i in {1,...,4}
{
\path[name intersections={of=root_line and line\i,by=L\i}];
\fill[blue] (L\i) circle (2pt);
\node[below left=0.1cm] at (L\i) {$L_\i$};
}

\draw (L1)+(1cm,0) arc (0:35:1cm) node[midway, right] {$\theta_1$};
\draw (L2)+(0.7cm,0) arc (0:60:0.7cm) node[midway, right] {$\theta_2$};
\draw (L3)+(0.5cm,0) arc (0:85:0.5cm) node[midway, right] {$\theta_3$};
\draw (L4)+(0.5cm,0) arc (0:140:0.5cm) node[midway, right=0.1cm] {$\theta_4$};

\draw [blue, ->] (0.5,1.5) -- node[left] {$t$} (0.5,3.5);

\end{tikzpicture} 

%% file: comp2.tex
\begin{tikzpicture}

\coordinate (P1) at (0, 0);
\coordinate (P2) at (0.5, 2);
\coordinate (P3) at (1, 3);
\coordinate (P4) at (4, 5);
\coordinate (P5) at (5, 6);

\coordinate (Q1) at (7, 0);
\coordinate (Q2) at (8, 3);
\coordinate (Q3) at (10, 4.5);
\coordinate (Q4) at (11, 6);

\coordinate (R1) at (12, 0);
\coordinate (R2) at (12.5, 2);
\coordinate (R3) at (13, 3);

\coordinate (R4) at (13.75, 3.33);
\coordinate (R5) at (14.25, 4.5);

\coordinate (R6) at (16, 5);
\coordinate (R7) at (17, 6);

\draw[thick] (P1)--(P2)--(P3)--(P4)--(P5); 
\foreach \i in {1,...,5}
{
\fill[black] (P\i) circle (3pt);
\node[above left] at (P\i) {$P_\i$};
}

\draw[thick, red] (P3)--(P4);
\node at (6,3) {\huge$\circ_3$};

\draw[thick] (Q1)--(Q2)--(Q3)--(Q4); 
\foreach \i in {1,...,4}
{
\fill[black] (Q\i) circle (3pt);
\node[above left] at (Q\i) {$Q_\i$};
}
\node at (11,3) {\huge$=$};

\draw[thick] (R1)--(R2)--(R3)--(R4)--(R5)--(R6)--(R7); 
\foreach \i in {1,...,7}
{
\fill[black] (R\i) circle (3pt);
\node[above left] at (R\i) {$R_\i$};
}
\draw[thick, red] (R3)--(R4)--(R5)--(R6); 

\end{tikzpicture} 

%% file: comp1.tex
\begin{tikzpicture}



\coordinate (P1) at (-0.5, 0);
\coordinate (P2) at (2, 0);
\coordinate (P3) at (4.5, 0);
\coordinate (P4) at (5.5, 0);

\coordinate (Q1) at (14.5, 0);
\coordinate (Q2) at (17, 0);
\coordinate (Q3) at (18.2, 0);
\coordinate (Q4) at (19.5, 0);
\coordinate (Q5) at (20.5, 0);

\coordinate (R1) at (8.5, 0);
\coordinate (R2) at (10, 0);
\coordinate (R3) at (12.5, 0);

\clip (-1,-0.7) rectangle (23,5.2);


\draw [->,shorten >= -1cm, dashed, blue, very thick, name path=root_line] (-1,0) -- (P4);
\draw[very thick, shorten <= -1.5cm, name path=line1] (P1)  -- +(35: 10cm);
\draw[very thick, shorten <= -1.5cm, name path=line2] (P2) -- +(50: 8cm);
\draw[very thick,  shorten <= -1.5cm, name path=line3] (P3) -- +(85: 7cm);
\draw[very thick, shorten <= -1.5cm, name path=line4] (P4) -- +(140: 7cm);

\foreach \i in {1,...,4}
{
\path[name intersections={of=root_line and line\i,by=L\i}];
\fill[blue] (L\i) circle (4pt);
\node[below right, font=\huge] at (L\i) {$P_\i$};
}

\path[name intersections={of=line2 and line3,by=L23}];
\fill[red] (L23) circle (4pt);
\node[right, font=\huge] at (L23) {$P_{23}$};
\fill[red] (L2)circle (4pt);
\fill[red] (L3)circle (4pt);

\node at (7.5,2) {\huge$*_2$};

\draw [->,shorten >= -0.5cm, dashed, blue, very thick, name path=root_line] (8,0) -- (R3);
\draw[very thick, shorten <= -1.5cm, name path=line1] (R1)  -- +(40: 7cm);
\draw[very thick, shorten <= -1.5cm, name path=line2] (R2) -- +(95: 8cm);
\draw[very thick, shorten <= -1.5cm, name path=line3] (R3) -- +(135: 7cm);

\foreach \i in {1,...,3}
{
\path[name intersections={of=root_line and line\i,by=L\i}];
\fill[blue] (L\i) circle (4pt);
\node[below right, font=\huge] at (L\i) {$Q_\i$};
}
\path[name intersections={of=line1 and line3,by=L13}];
\fill[red] (L13)circle (4pt);
\node[right, font=\huge] at (L13) {$Q_{13}$};
\fill[red] (L1)circle (4pt);
\fill[red] (L3)circle (4pt);

\node at (13.5,2) {\huge$=$};

\draw [->,shorten >= -1.5cm, dashed, blue, very thick, name path=root_line] (Q1)++(-1,0) -- (Q5);
\draw[very thick, shorten <= -1.5cm, name path=line1] (Q1)  -- +(35: 10cm);
\draw[very thick, red, shorten <= -1.5cm, name path=line2] (Q2) -- +(50: 8cm);
\draw[very thick, red, shorten <= -1.5cm, name path=line3] (Q3) -- +(70: 8cm);
\draw[very thick, red, shorten <= -1.5cm, name path=line4] (Q4) -- +(85: 7cm);
\draw[very thick, shorten <= -1.5cm, name path=line5] (Q5) -- +(140: 7cm);

\foreach \i in {1,...,5}
{
\path[name intersections={of=root_line and line\i,by=L\i}];
\fill[blue] (L\i) circle (4pt);
\node[below right, font=\huge] at (L\i) {$\i$};
}
\path[name intersections={of=line2 and line4,by=L24}];
\fill[red] (L24)circle (4pt);
\fill[red] (L2)circle (4pt);
\fill[red] (L4)circle (4pt);

\end{tikzpicture} 

%% file: comp3.tex
\begin{tikzpicture}

\coordinate (P1) at (0, 0);
\coordinate (P2) at (6, 0);
\coordinate (P3) at (0, 6);

\coordinate (Q1) at (12, 0);
\coordinate (Q2) at (14, 0);
\coordinate (Q3) at (16, 0);
\coordinate (Q4) at (18, 0);

\clip (-1,-0.7) rectangle (19,9);


\draw [->,shorten >= -1cm, dashed, blue, very thick, name path=root_line] (-1,0) -- (P2);
 \draw[draw=none, shorten <= -1.5cm, shorten >= -1.5cm, name path=line1] (P1)  -- node[below, font=\Huge] {$\tau_0$}  (P2);
\draw[very thick, shorten <= -1.5cm, shorten >= -1.5cm, name path=line2] (P1) -- node[left, font=\Huge] {$\tau_L$} (P3);
\draw[very thick,  shorten <= -1.5cm, shorten >= -1.5cm, name path=line3] (P3) -- node[right, font=\Huge] {$\tau_R$} (P2);

\fill[red] (P1) circle (4pt) node[below right,font=\huge, text=black] {$1$};
\fill[red] (P2) circle (4pt) node[below right,font=\huge, text=black] {$2$};
\fill[red] (P3) circle (4pt);

\draw [->,shorten >= -1.5cm, dashed, blue, very thick, name path=root_line] (Q1)++(-1,0) -- (Q4);
\draw[very thick, shorten <= -1.5cm, name path=line1] (Q1)  -- +(90: 10cm);
\draw[very thick, red, shorten <= -1.5cm, name path=line2] (Q2) -- +(100: 10cm);
\draw[very thick, red, shorten <= -1.5cm, name path=line3] (Q3) -- +(125: 11cm);
\draw[very thick, shorten <= -1.5cm, name path=line4] (Q4) -- +(135: 13cm);

\foreach \i in {1,...,4}
{
\path[name intersections={of=root_line and line\i,by=L\i}];
\fill[blue] (L\i) circle (4pt);
\node[below right, font=\huge] at (L\i) {$\i$};
}
\path[name intersections={of=line2 and line3,by=L23}];
\path[name intersections={of=line4 and line3,by=L34}];
\fill[red] (L23)circle (4pt);
\fill[black] (L34)circle (4pt);
\fill[red] (L2)circle (4pt);
\fill[red] (L3)circle (4pt);

\draw[->,>=Latex,thick](5,3) -- node[above, font=\Huge] {$T_i$} (9,3);

\end{tikzpicture} 

%% file: yb.tex
\begin{tikzpicture}


\coordinate (P1) at (0.5, 0);
\coordinate (P2) at (3, 0);
\coordinate (P3) at (4.5, 0);

\coordinate (Q1) at (14.5, 0);
\coordinate (Q2) at (16, 0);
\coordinate (Q3) at (18.5, 0);

\coordinate (R1) at (8.5, 0);
\coordinate (R2) at (10, 0);
\coordinate (R3) at (11.5, 0);

\clip (-1,-0.5) rectangle (19,4);


\draw [->,shorten >= -1cm, dashed, blue, very thick, name path=root_line] (0,0) -- (P3);
\draw[very thick, shorten <= -1.5cm, name path=line1] (P1)  -- +(50: 10cm);
\draw[very thick, shorten <= -1.5cm, name path=line2] (P2) -- +(95: 8cm);
\draw[very thick, shorten <= -1.5cm, name path=line3] (P3) -- +(145: 7cm);

\foreach \i in {1,...,3}
{
\path[name intersections={of=root_line and line\i,by=L\i}];
\fill[blue] (L\i) circle (3pt);
\node[above right] at (L\i) {\Huge$\i$};
}

\draw [->,shorten >= -0.5cm, dashed, blue, very thick, name path=root_line] (8,0) -- (R3);
\draw[very thick, shorten <= -1.5cm, name path=line1] (R1)  -- +(50: 10cm);
\draw[very thick, shorten <= -1.5cm, name path=line2] (R2) -- +(90: 8cm);
\draw[very thick, shorten <= -1.5cm, name path=line3] (R3) -- +(130: 7cm);

\foreach \i in {1,...,3}
{
\path[name intersections={of=root_line and line\i,by=L\i}];
\fill[blue] (L\i) circle (3pt);
\node[above right] at (L\i) {\Huge$\i$};
}

\draw [->,shorten >= -0.5cm, dashed, blue, very thick, name path=root_line] (14,0) -- (Q3);
\draw[very thick, shorten <= -1.5cm, name path=line1] (Q1)  -- +(40: 10cm);
\draw[very thick, shorten <= -1.5cm, name path=line2] (Q2) -- +(95: 8cm);
\draw[very thick, shorten <= -1.5cm, name path=line3] (Q3) -- +(135: 7cm);

\foreach \i in {1,...,3}
{
\path[name intersections={of=root_line and line\i,by=L\i}];
\fill[blue] (L\i) circle (3pt);
\node[above right] at (L\i) {\Huge$\i$};
}

\end{tikzpicture} 

%% file: comp4.tex
\begin{tikzpicture}



\coordinate (P1) at (-0.5, 0);
\coordinate (P2) at (2, 0);
\coordinate (P3) at (4.5, 0);
\coordinate (P4) at (5.5, 0);

\clip (-1,-0.7) rectangle (23,5.2);


\coordinate (ZP) at (3.5, 4);
\draw [->,shorten >= -1cm, dashed, blue, very thick, name path=root_line] (-1,0) -- (P4);
\draw[very thick, shorten <= -1.5cm, shorten >= -1.5cm, name path=line1] (P1)  -- (ZP);
\draw[very thick, shorten <= -1.5cm, shorten >= -1.5cm,name path=line2] (P2) -- (ZP);
\draw[very thick,  shorten <= -1.5cm, shorten >= -1.5cm,name path=line3] (P3) -- (ZP);
\draw[very thick, shorten <= -1.5cm, shorten >= -1.5cm, name path=line4] (P4) -- (ZP);

\foreach \i in {1,...,4}
{
\path[name intersections={of=root_line and line\i,by=L\i}];
\fill[blue] (L\i) circle (4pt);
\node[below right, font=\huge] at (L\i) {$P_\i$};
}

\path[name intersections={of=line2 and line3,by=L23}];
\fill[red] (L23) circle (4pt);
\node[right, font=\huge] at (L23) {$A$};
\fill[red] (L2)circle (4pt);
\fill[red] (L3)circle (4pt);

\node at (7.5,2) {\huge$*_2$};

\coordinate (R1) at (8.5, 0);
\coordinate (R2) at (10, 0);
\coordinate (R3) at (12.5, 0);

\coordinate (ZR) at (10, 4);
\draw [->,shorten >= -0.5cm, dashed, blue, very thick, name path=root_line] (8,0) -- (R3);
\draw[very thick, shorten <= -1.5cm, shorten >= -1.5cm, name path=line1] (R1)  -- (ZR);
\draw[very thick, shorten <= -1.5cm, shorten >= -1.5cm,name path=line2] (R2) -- (ZR);
\draw[very thick,  shorten <= -1.5cm, shorten >= -1.5cm,name path=line3] (R3) -- (ZR);

\foreach \i in {1,...,3}
{
\path[name intersections={of=root_line and line\i,by=L\i}];
\fill[blue] (L\i) circle (4pt);
\node[below right, font=\huge] at (L\i) {$Q_\i$};
}
\path[name intersections={of=line1 and line3,by=L13}];
\fill[red] (L13)circle (4pt);
\node[right, font=\huge] at (L13) {$A$};
\fill[red] (L1)circle (4pt);
\fill[red] (L3)circle (4pt);

\node at (13.5,2) {\huge$=$};


\coordinate (Q1) at (14.5, 0);
\coordinate (Q2) at (17, 0);
\coordinate (Q3) at (18, 0);
\coordinate (Q4) at (19.5, 0);
\coordinate (Q5) at (20.5, 0);

\coordinate (ZQ) at (18.5, 4);
\draw [->,shorten >= -1.5cm, dashed, blue, very thick, name path=root_line] (Q1)++(-1,0) -- (Q5);
\draw[very thick, shorten <= -1.5cm, shorten >= -1.5cm, name path=line1] (Q1)  -- (ZQ);
\draw[very thick, shorten <= -1.5cm, shorten >= -1.5cm,name path=line2] (Q2) -- (ZQ);
\draw[very thick,  shorten <= -1.5cm, shorten >= -1.5cm,name path=line3] (Q3) -- (ZQ);
\draw[very thick, shorten <= -1.5cm, shorten >= -1.5cm,name path=line4] (Q4) -- (ZQ);
\draw[very thick,  shorten <= -1.5cm, shorten >= -1.5cm,name path=line5] (Q5) -- (ZQ);

\foreach \i in {1,...,5}
{
\path[name intersections={of=root_line and line\i,by=L\i}];
\fill[blue] (L\i) circle (4pt);
}
\path[name intersections={of=line2 and line4,by=L24}];
\fill[red] (L24)circle (4pt);
\fill[red] (L2)circle (4pt);
\fill[red] (L3)circle (4pt);
\fill[red] (L4)circle (4pt);

\node[right, font=\huge] at (L24) {$A$};

\end{tikzpicture} 

%% file: hex.tex
\begin{tikzpicture}
   \newdimen\R
   \R=2.7cm

     \draw[->, red, line width=3pt] (180:\R) \foreach \x in {120,60,0} {  -- (\x:\R) };
     \draw[->, green, line width=3pt] (180:\R) \foreach \x in {240,300,0} {  -- (\x:\R) };
     \draw[->, cyan, line width=3pt] (180: \R) -- (0: \R);

    \draw (0:\R) 
    \foreach \x/\l/\p in 
    { 60/{$\{2,3\}$,1}/right,
      120/{2,$\{1,3\}$}/above,
      180/{$\{1,2\}$,3}/left,
      240/{$\{1,3\}$,2}/left,
      300/{3,$\{1,2\}$}/below,
      360/{3,$\{1,2\}$}/right
     }
    {  -- (\x:\R) node[midway,label={\p:\l}] {}};
   \foreach \x/\l/\p in
     { 60/{2,3,1}/above,
      120/{2,1,3}/above,
      180/{1,2,3}/left,
      240/{1,3,2}/below,
      300/{3,1,2}/below,
      360/{3,2,1}/right
     }
     \node[inner sep=1pt,circle,draw,fill,label={\p:\l}] at (\x:\R) {};

     \node at (0,0) {$\{1,2,3\}$};

\end{tikzpicture}

%% file: strips.tex
\begin{tikzpicture}


\coordinate (P1) at (2, 0);
\coordinate (P2) at (4.5, 0);
\coordinate (P3) at (6.5, 0);
\coordinate (P4) at (8, 0);

\clip (-1,0) rectangle (10,6);


\draw [->,shorten >= -1.5cm, dashed, blue, thick, name path=root_line] node[left] {$\ell_0$} (0,0) -- (P4);

\draw[thick, shorten <= -1.5cm, name path=line1] (P1)  -- +(35: 10cm);
\draw[thick, shorten <= -1.5cm, name path=line2] (P2) -- +(60: 8cm);
\draw[thick, shorten <= -1.5cm, name path=line3] (P3) -- +(80: 8cm);
\draw[thick, shorten <= -1.5cm, name path=line4] (P4) -- +(140: 8cm);

\draw[line width=8pt, cyan, shorten <= -1.5cm, name intersections={of=line1 and line3}] (P1)  -- (intersection-1);
\draw[line width=8pt, cyan, shorten <= -1.5cm, name intersections={of=line2 and line3}] (P2)  -- (intersection-1);
\draw[line width=8pt, cyan, shorten <= -1.5cm, name intersections={of=line3 and line2}] (P3)  -- (intersection-1);
\draw[line width=8pt, cyan, shorten <= -1.5cm, name intersections={of=line1 and line4}] (P4)  -- (intersection-1);

\draw[thick, shorten <= -1.5cm, name path=line1] (P1)  -- +(35: 10cm);
\draw[thick, shorten <= -1.5cm, name path=line2] (P2) -- +(60: 8cm);
\draw[thick, shorten <= -1.5cm, name path=line3] (P3) -- +(80: 8cm);
\draw[thick, shorten <= -1.5cm, name path=line4] (P4) -- +(140: 8cm);

\path[name path=upper_line] (0,5)--(10,5);
 
\draw [blue, ->] (0.5,1.5) -- node[left] {$t$} (0.5,3.5);

\end{tikzpicture} 

%% file: main.bib
@book{LodayVallette,
  author    = {Jean-Louis Loday and Bruno Vallette},
  title     = {Algebraic Operads},
  publisher = {Springer},
  year      = {2012},
  series    = {Grundlehren der mathematischen Wissenschaften},
  volume    = {346},
  address   = {Berlin, Heidelberg},
  isbn      = {978-3-642-30361-6} 
}

@book{Giraudo,
  title={Nonsymmetric operads in combinatorics},
  author={Giraudo, Samuele},
  year={2018},
  publisher={Springer International Publishing},
  volume={193},
  series={Lecture Notes in Mathematics},
  isbn={978-3-030-02074-3}
}

@book{MarklShniderStasheff,
  title={Operads in algebra, topology and physics},
  author={Markl, Martin and Shnider, Steve and Stasheff, Jim},
  year={2002},
  publisher={American Mathematical Society},
  series={Mathematical Surveys and Monographs},
  volume={96},
}

@article{Moser,
  title={Three integrable Hamiltonian systems connected with isospectral deformations},
  author={Moser, J{\"u}rgen},
  journal={Advances in Mathematics},
  volume={16},
  number={2},
  pages={197--220},
  year={1975},
  publisher={Elsevier}
}

@article{Kulish,
  title={Factorization of the classical and the quantum S matrix and conservation laws},
  author={Kulish, P. P.},
  journal={Theoretical and Mathematical Physics},
  volume={26},
  pages={132--137},
  year={1976},
  publisher={Springer}
}

@book{Perelomov,
  title={Integrable systems of classical mechanics and Lie algebras. Vol. 1},
  author={Perelomov, A. M.},
  series={Birkh{\"a}user texts in physics},
  year={1990},
  publisher={Birkh{\"a}user},
}

@article{OlshanetskyPerelomov,
  title={Classical integrable finite-dimensional systems related to Lie algebras},
  author={Olshanetsky, M. A. and Perelomov, A. M.},
  journal={Physics Reports},
  volume={71},
  number={5},
  pages={313--400},
  year={1981},
  publisher={Elsevier}
}

@article{KKS,
  title={Hamiltonian group actions and dynamical systems of Calogero type},
  author={Kazhdan, David and Kostant, Bertram and Sternberg, Shlomo},
  journal={Communications on Pure and Applied Mathematics},
  volume={31},
  number={4},
  pages={481--507},
  year={1978},
  publisher={Wiley}
}

@book{Etingof,
  title={Calogero-Moser systems and representation theory},
  author={Etingof, Pavel I.},
  series={Zurich Lectures in Advanced Mathematics},
  year={2014},
  publisher={European Mathematical Society}
}

@article{Zamolodchikov77,
  title={Exact two-particle S-matrix of quantum sine-Gordon solitons},
  author={Zamolodchikov, Al B},
  journal={Communications in Mathematical Physics},
  volume={55},
  pages={183--186},
  year={1977},
  publisher={Springer}
}

@article{Zamolodchikovs,
  title={Factorized S-matrices in two dimensions as the exact solutions of certain relativistic quantum field theory models},
  author={Zamolodchikov, Alexander B and Zamolodchikov, Alexey B},
  journal={Annals of physics},
  volume={120},
  number={2},
  pages={253--291},
  year={1979},
  publisher={Elsevier}
}

@misc{Torielli,
  title={LonTI Lectures on Sine-Gordon and Thirring},
  author={Torrielli, Alessandro},
  year={2022},
  eprint={2211.01186},
  archivePrefix={arXiv},
  primaryClass={hep-th}
}

@incollection{Dorey,
  title={Exact S-matrices},
  author={Dorey, Patrick},
  booktitle={Conformal Field Theories and Integrable Models: Lectures Held at the E{\"o}tv{\"o}s Graduate Course, Budapest, Hungary, 13--18 August 1996},
  pages={85--125},
  year={2007},
  publisher={Springer}
}

@book{Arutyunov,
   author = {{Arutyunov}, Gleb},
    title = "{Elements of Classical and Quantum Integrable Systems}",
     year = 2019,
      doi = {10.1007/978-3-030-24198-8},  
}

@article{Bashk,
  author       = {Denis Bashkirov},
  title        = {Lattice operads and operad filtrations},
  year         = {2024},
  version      = {v3},
  eprint       = {2212.14833},
  archivePrefix= {arXiv},
  primaryClass = {math.RA},
  url          = {https://arxiv.org/abs/2212.14833},
  note         = {arXiv:2212.14833 [math.RA]}
}

@article{Giraudo2,
  title={Combinatorial operads from monoids},
  author={Giraudo, Samuele},
  journal={Journal of Algebraic Combinatorics},
  volume={41},
  pages={493--538},
  year={2015},
  publisher={Springer}
}

@article{Wahl,
  title={Framed discs operads and Batalin--Vilkovisky algebras},
  author={Salvatore, Paolo and Wahl, Nathalie},
  journal={Quarterly Journal of Mathematics},
  volume={54},
  number={2},
  pages={213--231},
  year={2003},
  publisher={OUP}
}

@book{Yau,
  title={Infinity operads and monoidal categories with group equivariance},
  author={Yau, Donald},
  year={2021},
  publisher={World Scientific}
}

@article{Reid,
  title={Tiling with similar polyominoes},
  author={Reid, Michael},
  journal={Journal of Recreational Mathematics},
  volume={31},
  number={1},
  pages={15--24},
  year={2003},
  publisher={BAYWOOD PUBLISHING COMPANY, INC.}
}

@article{Golomb,
  title={Tiling with polyominoes},
  author={Golomb, Solomon W},
  journal={Journal of Combinatorial Theory},
  volume={1},
  number={2},
  pages={280--296},
  year={1966},
  publisher={Academic Press}
}

@misc{Khovanova,
  title = {L-Reptiles},
  author = {Khovanova, Tanya},
  year = {2010},
  howpublished = {\url{https://blog.tanyakhovanova.com/2010/04/l-reptiles/}},
  note = {Accessed: 2023-12-19}
}

@article{Lambrechts,
  title={Associahedron, cyclohedron and permutohedron as compactifications of configuration spaces},
  author={Lambrechts, Pascal and Turchin, Victor and Voli{\'c}, Ismar},
  journal={Bulletin of the Belgian Mathematical Society-Simon Stevin},
  volume={17},
  number={2},
  pages={303--332},
  year={2010},
  publisher={The Belgian Mathematical Society}
}

@article{Olver,
  title={Moving frames},
  author={Olver, Peter J},
  journal={Journal of Symbolic Computation},
  volume={36},
  number={3-4},
  pages={501--512},
  year={2003},
  publisher={Elsevier}
}

@article{Voronov,
  title={The Swiss-cheese operad},
  author={Voronov, Alexander A},
  journal={Contemporary Mathematics},
  volume={239},
  pages={365--374},
  year={1999},
  publisher={Providence, RI: American Mathematical Society}
}
